\documentclass[journal,onecolumn]{IEEEtran}

\usepackage{graphicx}
\usepackage{amsfonts}
\usepackage{amsmath,amsthm}
\usepackage[T1]{fontenc}
\usepackage[utf8]{inputenc}
\usepackage{xcolor}
\newcommand{\media}{\mathfrak{M}}

\newtheorem{theorem}{Theorem}
\newtheorem{lemma}{Lemma}
\newtheorem{proposition}{Proposition}

\newtheorem{remark}{Remark}
\newtheorem{example}{Example}

\begin{document}

\title{{Metric properties of homogeneous and spatially inhomogeneous $F$-divergences}
\author{Nicolò De Ponti}
\thanks{N. De Ponti is with the Department of Mathematics, University of Pavia, Pavia 27100, Italy (e-mail: nicolo.deponti01@universitadipavia.it)}}

\IEEEpeerreviewmaketitle

\maketitle

\begin{abstract}
In this paper I investigate the construction and the properties of the so-called \textit{marginal perspective cost} $H$, a function related to Optimal Entropy-Transport problems obtained by a minimizing procedure, involving a cost function $c$ and an entropy function. In the pure entropic case, which corresponds to the choice $c=0$, the function $H$ naturally produces a symmetric divergence. I consider various examples of entropies and I compute the induced marginal perspective function, which includes some well-known functionals like the Hellinger distance, the Jensen-Shannon divergence and the Kullback-Liebler divergence. I discuss the metric properties of these functions and I highlight the important role of the so-called Matusita divergences. In the entropy-transport case, starting from the power like entropy $F_p(s)=(s^p-p(s-1)-1)/(p(p-1))$ and the cost $c=d^2$ for a given metric $d$, the main result of the paper ensures that for every $p>1$ the induced marginal perspective cost $H_p$ is the square of a metric on the corresponding cone space.
\end{abstract}


\IEEEoverridecommandlockouts

\section{Introduction}
\IEEEPARstart{G}{iven} a function $F\in \Gamma_0(\mathbb{R}_{+}):=\{f:[0,+\infty)\rightarrow [0,+\infty], f \ \mathrm{convex,\ lower \ semicontinuous \ and} \
 f(1)=0\}$, a finite set $\Omega=\{x_1,..,x_m\}$, and two probability densities 
$$\mu_1=\sum_{i=1}^mr_i\delta_{x_i}, \ \ \mu_2=\sum_{i=1}^mt_i\delta_{x_i}$$
such that $t_i>0$ when $r_i>0$ for every $i=1,..,m,$ the $F$-divergence of $\mu_1$ from $\mu_2$ is defined as
\begin{equation}
D_F(\mu_1||\mu_2):=\sum_{i=1}^mF\Big(\frac{r_i}{t_i}\Big)t_i=\sum_{i=1}^m\hat{F}(r_i,t_i)
\end{equation}
where $\hat{F}(r,t):=F\big(\frac{r}{t}\big)t$ is the perspective function induced by $F$ (here I am using the convention $F\big(\frac{0}{0}\big)0=0$).

Since their introduction by Csiszár \cite{Csiszar}, Ali and Silvey \cite{Ali}, $F$-divergences have become a fundamental tool in information theory and statistics. They can be interpreted as a sort of "distance function" on the set of probability distributions, even if they do not generally fulfill the symmetric property and the triangle inequality. I refer to Liese and Vajda \cite{Liese2}, \cite{Vajda}, and references therein for a systematic presentation of these functionals, including the \textit{total variation} (for $F(s)=|s-1|$), and the $\chi^{\alpha}$ \textit{divergences} generated by the choice $F(s)=|s-1|^{\alpha}$ (discussed by Vajda in \cite{Vajda2}). Another important class of divergences is represented by the so-called \textit{Matusita divergences} $F(s)=|s^a-1|^{\frac{1}{a}}$ \cite{Matusita}, which include as a particular case the well known \textit{Hellinger distance} $F(s)=(\sqrt{s}-1)^2$ \cite{Hellinger}.

Starting from a $F$-divergence, there is a simple variational way to generate a new symmetric divergence by setting
\begin{equation}
H_{F}(\mu_1||\mu_2):=\inf_{\mu} D_F(\mu||\mu_1)+D_F(\mu||\mu_2).
\end{equation}
This is related to the marginal perspective function $H$, the lower semicontinuous envelope of the function
\begin{equation}\label{minimizzazione tramite H, senza costo}
\tilde{H}(r,t)=\inf_{\theta >0}F\Big(\frac{\theta}{r}\Big)r+F\Big(\frac{\theta}{t}\Big)t=\inf_{\theta>0}\hat{F}(\theta,r)+\hat{F}(\theta,t).
\end{equation}
The function $H$ obtained in this way is jointly convex, lower semicontinuous and it is zero on the diagonal. As a result, one gets a natural map 
\begin{equation}
T_1:\Gamma_0({\mathbb{R}_{+}})\rightarrow \Gamma_0({\mathbb{R}_{+}}), \ \ \ \ T_1(F)(s):=H(1,s),
\end{equation}
with the additional property
\begin{equation}
D_{T_1(F)}(\mu_1||\mu_2)=D_{T_1(F)}(\mu_2||\mu_1).
\end{equation}

Using different functions $F\in \Gamma_0(\mathbb{R}_{+})$, that I also call \textit{entropy functions} in the present paper, the minimizing procedure \eqref{minimizzazione tramite H, senza costo} gives raise to well-known statistical functionals.

For the function $F(s)=U_1(s):=s\ln(s)-s-1$, the result is the Hellinger distance \cite{Hellinger}
\begin{equation}
H(r,t)=(\sqrt{r}-\sqrt{t})^2.
\end{equation}

When $F(s)=U_0(s):=s-1-\ln(s)$, one gets the Jensen-Shannon divergence \cite{Lin}
\begin{equation}
H(r,t)=r\ln(r)+t\ln(t)-(r+t)\ln\Big(\frac{r+t}{2}\Big).
\end{equation}

The previous examples are taken from the class of the power like entropies $\{U_p\}$
\begin{equation}
U_p(s):=\frac{1}{p(p-1)}(s^p-p(s-1)-1), \ \textrm{if} \ p\neq 0,1.
\end{equation}
They give raise to the family of functions
\begin{equation}\label{H di tipo potenza}
H(r,t)=\frac{2}{p}\Big[\mathfrak{M}_1(r,t)-\mathfrak{M}_{1-p}(r,t)\Big],
\end{equation}
where the expression is written in the terms of the power mean
\begin{equation}
\mathfrak{M}_p(r,t):=(\frac{r^p+t^p}{2})^{\frac{1}{p}}.
\end{equation}

The entropy $F(s)=s^2-2\ln(s)-1$ produces the symmetric Kullback-Leibler divergence \cite{Kullback}
\begin{equation}
H(r,t)=(r-t)\ln\Big(\frac{r}{t}\Big).
\end{equation}

The marginal perspective function can also be computed starting from non-smooth entropies as $F(s)=|s-1|$, which induces the celebrated total variation distance
\begin{equation}
H(r,t)=|r-t|.
\end{equation}

The metric properties of the $F$-divergences have been investigated by many authors like Csiszar, Endres, Kafka, Osterreicher, Schindelin, Vincze (\cite{Csiszar2}, \cite{Endres}, \cite{Kafka}, \cite{Osterreicher}, \cite{Osterreicher2}), to cite only a few.
In the pure entropic setting, I generalize a previous result of Osterreicher \cite{Osterreicher} and I prove that, for the power like entropy $U_p$, the induced function $H$ given by \eqref{H di tipo potenza} is the square of a metric on $[0,+\infty)$ for every $p\in (-\infty,\frac{1}{2}]\cup [1,+\infty).$

In the pure entropic case, I also characterize the limit of the sequence $T_1^{(n)}(F)$ and I prove that the total variation and its positive multiples are the only divergences that are also a distance. 
Under additional assumptions, the convergence properties of the sequence $T_a^{(n)}(F)$ are also studied, where I put $T_a(F):=2^{\frac{1}{a}-1}T_1(F)$, $a\in (0,1)$. I will show that this is strictly related to those divergences $F$ for which $H^a$ is a distance, and I will emphasize the central role of the class of Matusita divergences.

Recently, $F$-divergences have been considered by Liero, Mielke, Savaré \cite{LMS} as penalizing functionals in the formulation of Optimal Entropy-Transport problems, a generalization of Optimal-Transport problems obtained by relaxing the marginal constraints.
Given a cost function $c:X_1\times X_2\rightarrow [0,+\infty)$ and an admissible entropy function $F\in \Gamma_{0}(\mathbb{R}_{+})$, a crucial role in the theory is played by the induced \textit{marginal perspective cost} $H:X_1\times [0,+\infty)\times X_2\times[0,+\infty)\rightarrow [0,+\infty]$, the lower semicontinuos envelope of the function 

\begin{equation}
\label{minimizzazione tramite H}
\tilde{H}(x_1,r_1;x_2,r_2)=\inf_{\theta>0}r_1F\Big(\frac{\theta}{r_1}\Big)+r_2F\Big(\frac{\theta}{r_2}\Big)+\theta c(x_1,x_2).
\end{equation}

The function $H$ remains positively $1$-homogeneous with respect to $(r_1,r_2)$, a property used in \cite{LMS} in order to derive a "homogeneous formulation" of Optimal Entropy-Transport problems that allows the study of the metric and dynamical aspects of the theory. 

When the starting entropy $F$ has a strict minimum at $s=1$, and the cost $c$ is a symmetric function such that $c(x_1,x_2)=0$ if and only if $\ x_1=x_2$, I will show that the induced marginal perspective cost $H$ is symmetric, non-negative and $H(x_1,r_1;x_2,r_2)=0$ if and only if $(x_1,r_1)=(x_2,r_2)$ or $r_1=r_2=0$.

In the presence of a non-zero cost function $c$, an explicit computation of the induced marginal perspective cost is often unavailable. A special case, central in the study of Optimal Entropy-Transport problems, is given by the choices $X_1=X_2=X$, $c=d^2$ for a metric $d$ on $X$, and $F=U_p$. It holds

\begin{align}
&H_p(x_1,r;x_2,t)=\frac{2}{p}\Big[\mathfrak{M}_1(r,t)-\mathfrak{M}_{1-p}(r,t)\bigg(1+(1-p)\frac{d^2(x_1,x_2)}{2}\bigg)_{+}^{\frac{p}{p-1}}\Big], \ \ \ p\neq 0,1.
\end{align}

When $p=1$ or $p=0$, one gets

\begin{align}
&H_1(x_1,r;x_2,t)=2\Big[\mathfrak{M}_1(r,t)-\mathfrak{M}_0(r,t)e^{-\frac{d^2(x_1,x_2)}{2}}\Big],
\\
&H_0(x_1,r;x_2,t)=r\ln{r}+t\ln{t}-(r+t)\ln{\Big(\frac{r+t}{2+d^2(x_1,x_2)}\Big)}.
\end{align}

Our main theorem states that for any $p\geq 1$ the square root of $H_p$ satisfies the triangle inequality on the cone space over $X$. The latter is the space $\mathfrak{C}=Y/{\sim}$, where $Y=X\times [0,+\infty)$ and 
$$(x_1,r_1)\sim (x_2,r_2) \iff r_1=r_2=0 \ \mbox{or} \ r_1=r_2, x_1=x_2.$$
Thus, I provide new examples of entropy-transport metrics besides the \textit{Gaussian Hellinger-Kantorovich distance} ($p=1$) and the related \textit{Hellinger-Kantorovich distance} studied in \cite{LMS}. The class of examples includes, for $p=2$, a transport variant of the \textit{Vincze-Le Cam distance} \cite{Vincze}, \cite{LeCam},
\begin{equation}
H(r,t)=\frac{(r-t)^2}{2(r+t)}.
\end{equation}

This paper is organized as follows.

In Section II, I recall some basic concepts of convex analysis, in particular I discuss the connection between the entropy function and the induced perspective function.  

In the third section, I recall the definition of the power means and their main properties. The results in this section will be useful in the study of the marginal perspective cost generated by the power like entropies.

Section IV is devoted to the study of the costless version of the function $H$. I provide a list of examples of admissible entropy functions, which includes indicator functions, $\chi^{\alpha}$ divergences, Matusita divergences, power like entropies and other two families of convex functions that I have called \textit{power-logarithmic entropies} and \textit{double power entropies}. Then, I compute the induced marginal perspective function and I discuss the metric properties of the function obtained starting from some of the previous examples.
Finally, I study the convergence properties of the iteration of the minimizing procedure \eqref{minimizzazione tramite H, senza costo} and I will highlight the role of the class of Matusita divergences.

In the fifth section I introduce the notion of homogeneous marginal perspective cost and I discuss its main properties. 

In section VI, I present the Optimal Entropy-Transport problem and I briefly motivate the "homogeneous formulation" of this problem, via the homogeneous marginal perspective cost.

In the last section I focus on the marginal perspective cost $H_p$ induced by the power like entropy $U_p$ and by the cost $c=d^2$, for a given metric $d$. I prove the main theorem of the paper, which ensures that the function $H_p$ is the square of a metric on the corresponding cone space.

For the sake of simplicity, I limit the discussion to finite nonnegative measures over finite discrete set, but the results can be generalized to finite nonnegative Radon measures over Hausdorff topological spaces (see \cite{LMS}). I plan to address this case in a future work.

In this paper, a real function $f$ is increasing (resp. decreasing) if for any $r<s$ we have $f(r)\leq f(s)$ (resp. $f(r)\geq f(s)$).

\section{Entropy functions}
\label{Entropy function}

A function $F:[0,+\infty)\rightarrow [0,+\infty]$ belongs to the class $\Gamma_0({\mathbb{R}_{+}})$ of admissible entropy functions if $F$ is convex, lower semicontinuous and $F(1)=0$. The domain of the function $F$ is the set
\begin{equation} 
\mathrm{D}(F):=\big\{s\in [0,+\infty):F(s)<+\infty\big\}.
\end{equation}
Let $F\in\Gamma_0({\mathbb{R}_{+}})$, the recession function $\mathrm{rec}(F)$ and the recession constant $F^{'}_{\infty}$ are defined by
\begin{equation}\label{recession}
\mathrm{rec}(F)(r):=\lim_{\alpha \to +\infty}\frac{F(1+\alpha r)}{\alpha}, \ \ \ \ F^{'}_{\infty}:=\mathrm{rec}(F)(1).
\end{equation} 

The perspective function induced by $F$ is the function $\hat{F}:[0,+\infty)\times [0,+\infty)\rightarrow [0,+\infty]$, given by
\begin{equation}
\hat{F}(r,t):=\begin{cases} F\big(\frac{r}{t}\big)t &\mbox{if} \ t>0\\
\mathrm{rec}(F)(r) &\mbox{if} \ t=0.\end{cases}
\end{equation}
$\hat{F}$ is jointly convex, lower semicontinuous and $\hat{F}(r,r)=0$ for any $r$.

The right derivative $F'_{0}$ at $0$, and the asymptotic affine coefficient $\mbox{aff}F_{\infty}$ are defined by

\begin{equation}
F'_{0}:=\begin{cases}-\infty&\mbox{if} \ \ F(0)=+\infty,\\ 
\lim_{s\downarrow 0}\frac{F(s)-F(0)}{s} &\mbox{otherwise},\end{cases}
\end{equation}

\begin{equation}
\mbox{aff}F_{\infty}:=\begin{cases}+\infty&\mbox{if} \ \ F'_{\infty}=+\infty,\\ 
\lim_{s\to \infty}\big(F'_{\infty}s-F(s)\big)&\mbox{otherwise},\end{cases}
\end{equation}
 
which are well posed due to the convexity of $F$. 

The Legendre conjugate function $F^{*}: \mathbb{R}:\rightarrow (-\infty,+\infty]$ is defined by

\begin{equation}
F^{*}(\phi):=\sup_{s\geq 0}\{s\phi-F(s)\}.
\end{equation}

$F^{*}$ is the conjugate of the convex function $\tilde{F}:\mathbb{R}\rightarrow [0,+\infty]$ obtained by extending $F$ to $+\infty$ for negative arguments. It is convex and lower semicontinuous. 
Concerning the behavior of $F^{*}$, the following Lemma holds (\cite{LMS}, section $2.3$):
\begin{lemma}
\label{omeomorfismo coniugata}
The function $F^{*}$ is an increasing homeomorphism between $(F_0',F'_{\infty})$ and $(-F(0),\mathrm{aff}F_{\infty})$ with $F^{*}(0)=0$.
\end{lemma}

The reverse entropy function $R:[0,\infty)\rightarrow [0,\infty]$ is defined by
\begin{equation}
R(s):=\begin{cases}F(\frac{1}{s})s&\mbox{if} \ s>0\\ F^{'}_{\infty}&\mbox{if} \ s=0,\end{cases}
\end{equation}
so that $R(s)=\hat{F}(1,s).$ In particular, $R$ is convex, lower semicontinuous and the map $F\mapsto R$ is an involution of $\Gamma_0({\mathbb{R}_{+}})$. Moreover, it holds $\hat{F}(r,t)=\hat{R}(t,r)$ and the function $R$ satisfies

\begin{equation}
\label{coefficienti R}
\begin{split}
R(1)=0, \ R(0)=F'_{\infty}, \ R'_{\infty}=F(0), \\
R'_{0}=-\mathrm{aff}F_{\infty}, \ \mathrm{aff}R_{\infty}=-F'_{0}.
\end{split}
\end{equation}

Starting from a function $F\in \Gamma_0(\mathbb{R}_{+})$, a finite set $\Omega=\{x_1,..,x_m\}$, and two probability densities 
\begin{equation}
\mu_1=\sum_{i=1}^mr_i\delta_{x_i}, \ \ \mu_2=\sum_{i=1}^mt_i\delta_{x_i},\end{equation}
the $F$-divergence of $\mu_1$ from $\mu_2$ is given by
\begin{equation}
D_F(\mu_1||\mu_2):=\sum_{i=1}^m\hat{F}(r_i,t_i)=\sum_{i=1}^m\hat{R}(t_i,r_i).
\end{equation}

The Legendre conjugates of $F$ and $R$ are related by 
\begin{equation}
\psi\leq -F^{*}(\phi)\iff \phi\leq -R^{*}(\psi).
\end{equation}

\section{Power means}

In this section I study the \textit{power means} (also called \textit{generalized means}), a family of functions that includes the well-known arithmetic, geometric and harmonic means. The property of these functions will be useful later on.

In what follows $r,t$ will denote two non-negative real numbers and $p$ a real parameter, which I suppose for the present not to be $0$. 
The $p$-power mean between $r$ and $t$ is given by
\begin{equation}
\mathfrak{M}_p(r,t):=(\frac{r^p+t^p}{2})^{\frac{1}{p}},
\end{equation}
except when $p<0$ and $r$ or $t$ is zero. In this case $\media_p$ is equal to zero:
\begin{equation}
\mathfrak{M}_p(r,t)=0 \ \ \ \ (p<0, \ r=0 \ \mathrm{or} \ t=0).
\end{equation}

In the case $p=0$ I put
\begin{equation}
\media_0(r,t):=\sqrt{rt}
\end{equation} 
so that $\lim_{p \to 0} \media_p(r,t)=\media_0(r,t).$

It is easy to see that $\media_p(r,r)=r$ for every $p \in \mathbb{R}$ and every $r\geq 0$. The function $\media_p$ is symmetric, i.e. $\media_p(r,t)=\media_p(t,r),$
and positively $1$-homogeneous in the sense that $\media_p(\lambda r,\lambda t)=\lambda \media_p(r,t)$ for every $\lambda\geq 0.$ Moreover, it is not difficult to prove that $M_p(r,s)\leq M_p(r,t)$ for every $p$, $r$ and $s\leq t.$

$\media_1$ is the well-known arithmetic mean, $\media_0$ is the geometric mean and $\media_{-1}$ is called harmonic mean. 

The main theorem (see $\cite{Bullen}$ for a proof) regarding the power means is the following:

\begin{theorem}
\label{proprieta' medie}
If $p_1<p_2$ then 
$$\media_{p_1}(r,t)\leq\media_{p_2}(r,t)$$
with the case of equality given by $r=t$, or $p_2\leq 0$ and $r\wedge t=0$.

\end{theorem}

In particular,
\begin{equation}
r \wedge t=\lim_{p \to -\infty}\media_p(r,t)\leq \media_{p}(r,t)\leq \lim_{p \to +\infty}\media_p(r,t)=r \vee t,
\end{equation}
for any $p\in \mathbb{R}$, $r,t \in [0,\infty).$

\section{Costless marginal perspective}
\label{sezione senza costo}

Let $F\in \Gamma_0(\mathbb{R}_{+})$ be an admissible entropy function and let $R$ be its reverse entropy. In general, for the induced perspective function one has $\hat{F}\neq \hat{R}$, so that the $F$-divergence does not satisfy the symmetric property. In order to replace $F$ with a new "symmetric entropy", a natural procedure is the following: define the marginal perspective function $H_F:[0,+\infty)\times [0,+\infty)\rightarrow [0,+\infty]$ as the lower semicontinuous envelope of the function

\begin{equation}\label{minimizzazione senza costo}
\tilde{H}_F(r_1,r_2):=\inf_{\theta>0}\Big(R\big(\frac{r_1}{\theta}\big)+R\big(\frac{r_2}{\theta}\big)\Big)\theta.\\
\end{equation}

An equivalent definition can be given in term of the induced perspective functions $\hat{F}$ or $\hat{R}$ by:
\begin{align}\label{minimizzazione sensza costo, versione perspective}
\tilde{H}_F(r_1,r_2):&=\inf_{\theta>0}\hat{F}(\theta,r_1)+\hat{F}(\theta,r_2)\\
&=\inf_{\theta>0}\hat{R}(r_1,\theta)+\hat{R}(r_2,\theta)
\end{align}

The infimum in the definition is a minimum and it occurs in the interval $[r_1,r_2]$ (without loss of generality I am assuming $r_1\leq r_2$): to see this it is enough to notice that the function $\theta \mapsto \hat{F}(\theta,r_1)+\hat{F}(\theta,r_2)$ is lower semicontinuous and it is decreasing in $[0,r_1]$ and increasing in $[r_2,+\infty)$.
I will prove in section \ref{sezione H con costo} (in a more general context), that the function $H_F$ is non-negative, symmetric, jointly convex and positively $1$-homogeneous. Moreover, when the function $F$ has a strict minimum at $1$, $H_F(r,t)=0$ if and only if $r=t$. 
It is important to notice, since $H_F$ is $1$-homogeneous, that the study of the function $H_F$ is equivalent to the study of the $1$-variable function $s\mapsto H_F(1,s)\in \Gamma_0(\mathbb{R}_{+}).$ I will continuously use this fact in the paper. 

\subsection{Examples}
I consider now different examples of admissible entropy function $F$ and I compute the expression of the induced marginal perspective $H_F$. I will in general suppose $rt>0$, so that I can avoid ambiguous expressions at the boundary of the domain that should be treated carefully. 

\begin{example}(Indicator functions)
The \textit{indicator function} of the closed interval with endpoints $a$ and $b$, $0\leq a\leq 1\leq b\leq +\infty$, is defined by

\begin{equation}
I_{[a,b]}(s)=\begin{cases}0&\mbox{if} \ \ s\in [a,b],\\ +\infty &\mbox{if} \ \ s\not\in [a,b].\end{cases}
\end{equation}

When $F=I_{[a,b]}$ one obtains
\begin{equation}
\tilde{H}_{I_{[a,b]}}(r,t)=\begin{cases}0 &\mbox{if} \ \frac{a}{b}\leq \frac{r}{t}\leq \frac{b}{a},\\
+\infty &\mbox{otherwise}
,\end{cases}
\end{equation}
where $\frac{b}{a}=+\infty$ if $a=0$ and $\frac{a}{b}=0$ if $b=+\infty$.
\end{example}

\begin{example}($\chi^{\alpha}$ divergences)
Given a parameter $\alpha\geq 1$, the $\chi^{\alpha}$ \textit{divergence} is defined as
\begin{equation}
\chi^{\alpha}(s)=|s-1|^{\alpha}.
\end{equation}
$\chi^1=|s-1|$ is the famous total variation entropy.

The entropy function $F=\chi^{\alpha}$ gives raise to the marginal perspective function
\begin{equation}
\tilde{H}_{\chi^{\alpha}}(r,t)=\frac{|r-t|^{\alpha}}{(r+t)^{\alpha-1}}.
\end{equation}
We can recognize the expression of the so-called \textit{Puri-Vincze divergence}.
\end{example}

\begin{example}(Matusita divergences)
For $0<a\leq 1$ the Matusita divergence is given by $M_a(s)=|s^a-1|^{\frac{1}{a}}$.
Clearly $\chi^1=M_1.$

When $F=M_a$ it is easy to see that
\begin{equation}
H_{M_a}(r,t)=2^{1-\frac{1}{a}}|r^a-t^a|^{\frac{1}{a}}.
\end{equation}
It is interesting to note that except for the constant factor $2^{1-\frac{1}{a}}$, the Matusita function $M_a$ remains invariant after the minimizing procedure \eqref{minimizzazione senza costo}. I will come back to this point in section \ref{Sezione caratterizzazione universale}.
\end{example}

\begin{example}(Power like entropies)
Let $p$ be any real number. I call power-like entropy of order $p$ the function $U_p:[0,+\infty)\rightarrow [0,+\infty]$ characterized by 

\begin{equation}
U_p\in \mathcal{C}^{\infty}(0,+\infty), \,  U_p(1)=U'_p(1)=0, \, U_p''(s)=s^{p-2}, \, U_p(0):=\lim_{s\downarrow 0}U_p(s).
\end{equation}
The function $U_p$ can be computed explicitly and one gets: 

\begin{equation}
\begin{cases}U_p(s)=\frac{1}{p(p-1)}(s^p-p(s-1)-1)&\mbox{if} \ \ p\neq 0,1,\\ U_1(s)=s\ln(s)-s+1,\\ U_0(s)=s-1-\ln(s),\end{cases}
\end{equation}
with $\displaystyle U_p(0)=1/p$ for $p>0$ and $U_p(0)=+\infty$ for $p\leq 0$.
This family of functions, also called Dichotomy Class, was introduced by Liese and Vajda \cite{Liese},\cite{Vajda}.

Given $F=U_p$, we obtain the following expression:
\begin{equation}
\begin{cases}\tilde{H}_{U_p}(r,t)=\frac{1}{p}\Big[r+t-2^{\frac{p}{p-1}}(r^{1-p}+t^{1-p})^{\frac{1}{1-p}}\Big] \ \ p\neq 0,1,\\
\tilde{H}_{U_1}(r,t)=r+t-2\sqrt{rt},\\
\tilde{H}_{U_0}(r,t)=r\ln{r}+t\ln{t}-(r+t)\ln{\Big(\frac{r+t}{2}\Big)}
.\end{cases}
\end{equation}

We can recognize some well-known statistical functionals: for example in the logarithmic entropy case $p=1$ it appears the Hellinger distance 
\begin{equation}
H_{U_1}(r,t)=(\sqrt{r}-\sqrt{t})^2.
\end{equation}
I have already notice that the same function is obtained starting from the entropy $U_{\frac{1}{2}}(s)=2(\sqrt{s}-1)^2=2M_{\frac{1}{2}}$.

For $p=0$ we have the Jensen-Shannon divergence, a squared distance between measures derived from the Kullback-Leibler divergence (\cite{Endres}). 

The quadratic entropy $U_2(s)=\frac{1}{2}(s-1)^2$ gives raise to the triangular discrimination 
\begin{equation}
\tilde{H}_{U_2}(r,t)=\frac{1}{2}H_{\chi^2}=\frac{1}{2}\frac{(r-t)^2}{(r+t)}.
\end{equation}
\end{example}

\begin{example}(Power-logarithmic entropies)
Given a real number $p\geq 1$, I call power-logarithmic entropy of order $p$ the function $V_p:[0,+\infty)\rightarrow [0,+\infty]$ 

\begin{equation}
V_p(s):=s^p-p\ln(s)-1, \ \ \ s>0,
\end{equation}
and $V_p(0)=+\infty.$ It is easy to see that $V_p\in \mathcal{C}^{\infty}(0,+\infty)$ and $V_p(0)=\lim_{s\downarrow 0}V_p(s)$. 

Starting from the power-logarithmic entropy of order $p$ one gets:
\begin{equation}
\tilde{H}_{V_p}(r,t)=(r+t)\ln\Big[\frac{rt(r^{p-1}+t^{p-1})}{r+t}\Big]-p\big(r\ln(t)+t\ln(r)\big).
\end{equation}
As expected, $H_{V_1}=H_{U_0}$ since $V_1=U_0$.
When $p=2$, one obtains the symmetric Kullback-Leibler divergence \cite{Kullback}:
\begin{equation}
\tilde{H}_{V_2}(r,t)=(r-t)\ln\Big(\frac{r}{t}\Big).
\end{equation}

\end{example}

\begin{example}(Double power entropies)
Given two parameters $p,q$ such that $p\geq 1, \ 0<q \leq 1$ and $p\neq q$, or $p<0, \ q\geq 1$, the double power entropy of order $p,q$ is given by
\begin{equation}
W_{p,q}(s):=qs^p-ps^q+p-q, \ \ \ s>0.
\end{equation}
$W_{p,q}$ is a strictly convex function, $W_{p,q}\in \mathcal{C}^{\infty}(0,+\infty)$, and it is extendex in $0$ by continuity so that $W_{p,q}(0)=p-q$ when $p,q$ are positive, $W_{p,q}(0)=+\infty$ when $p<0.$

A direct computation shows that:
\begin{equation}
\tilde{H}_{W_{p,q}}(r,t)=(q-p)rt\bigg[\frac{(r^{q-1}+t^{q-1})^p}{(r^{p-1}+t^{p-1})^q}\bigg]^{\frac{1}{p-q}}-(q-p)(r+t).
\end{equation}
For example, when $p=3/2, q=1/2$ one gets
\begin{equation}
H_{W_{\frac{3}{2},\frac{1}{2}}}(r,t)=r+t-(rt)^{\frac{1}{4}}(\sqrt{r}+\sqrt{t}).
\end{equation}
\end{example}

\begin{figure}
\centering
\includegraphics[scale=0.34]{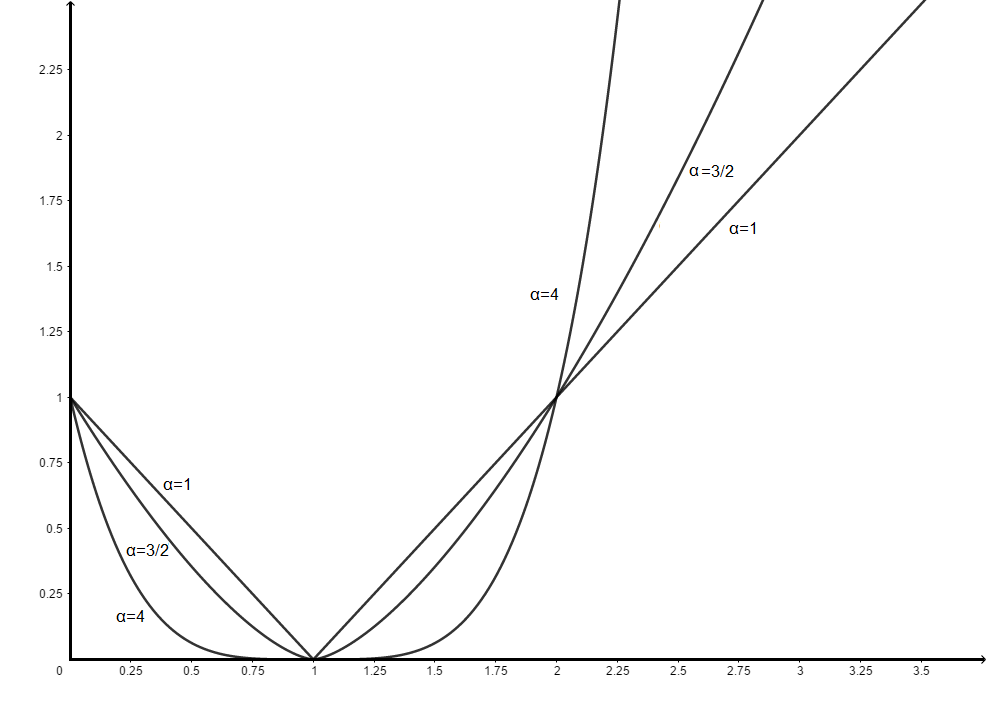}
\caption{$\chi^{\alpha}$ divergences}
\end{figure}

\begin{figure}
\centering
\includegraphics[scale=0.33]{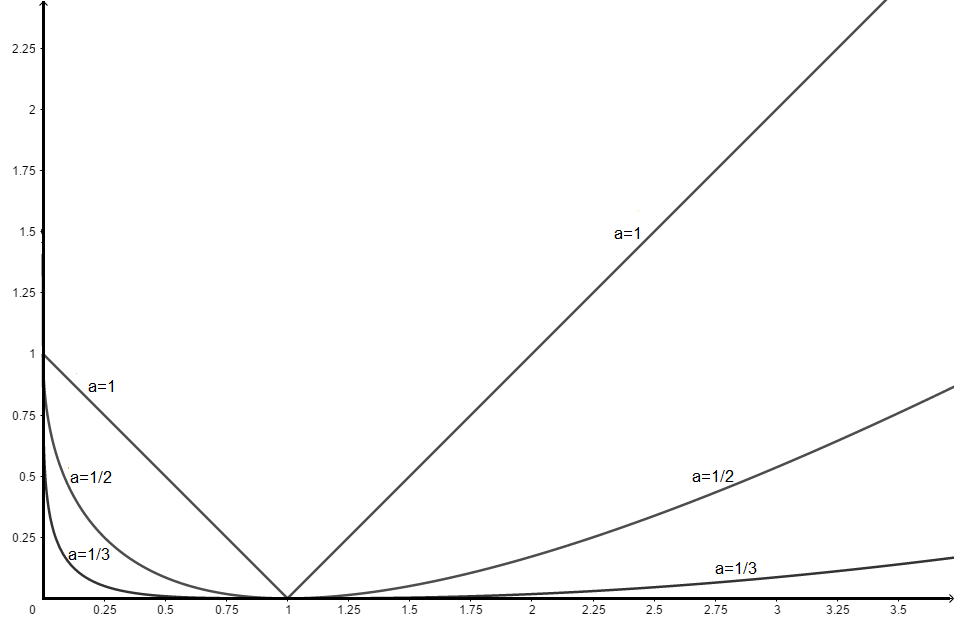}
\caption{Matusita divergences}
\end{figure}

\begin{figure}
\centering
\includegraphics[scale=0.35]{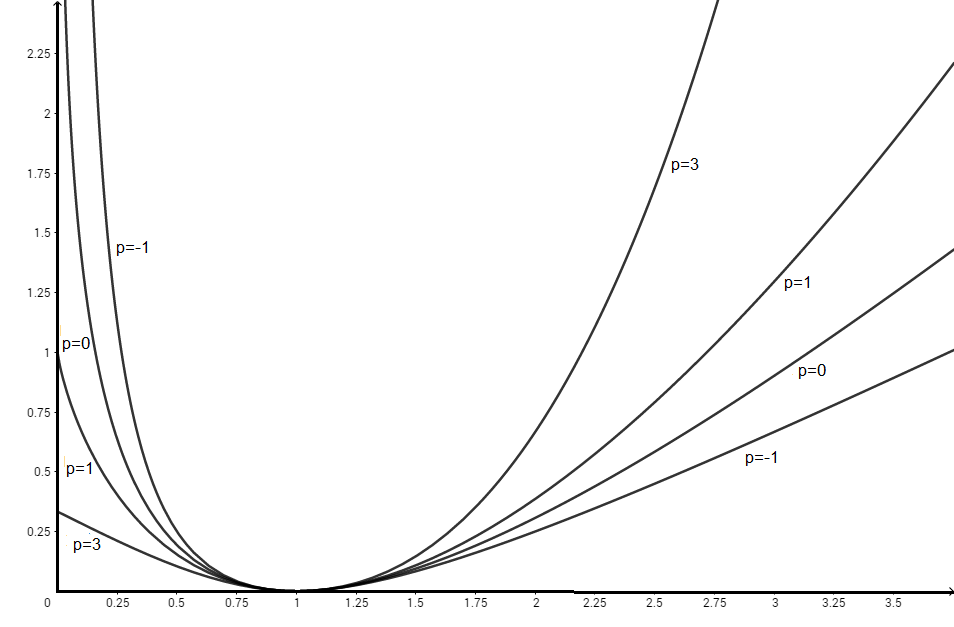}
\caption{Power-like entropies}
\end{figure}

\begin{figure}
\centering
\includegraphics[scale=0.34]{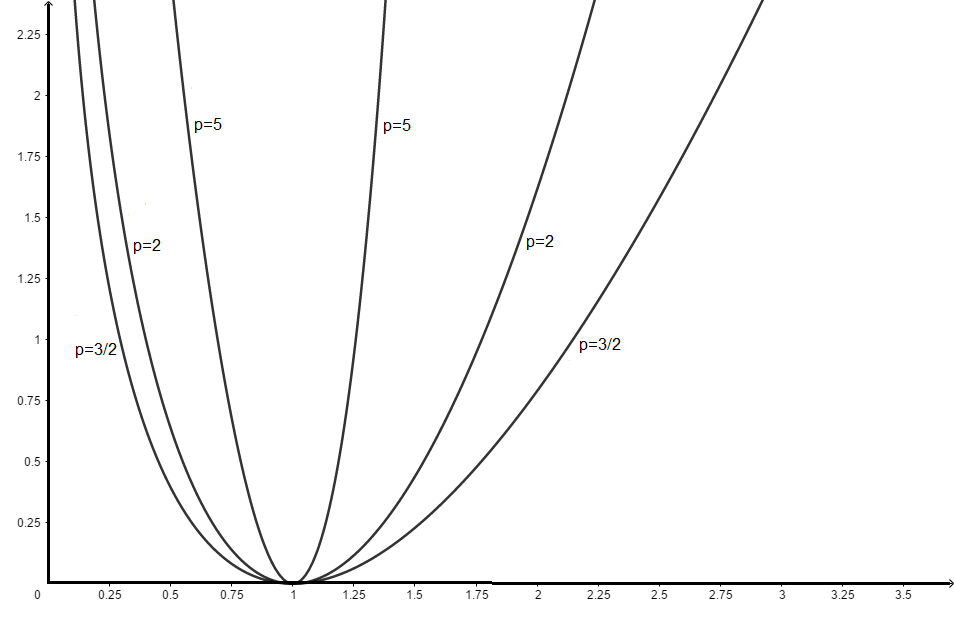}
\caption{Power logarithmic entropies}
\end{figure}

\begin{figure}
\centering
\includegraphics[scale=0.34]{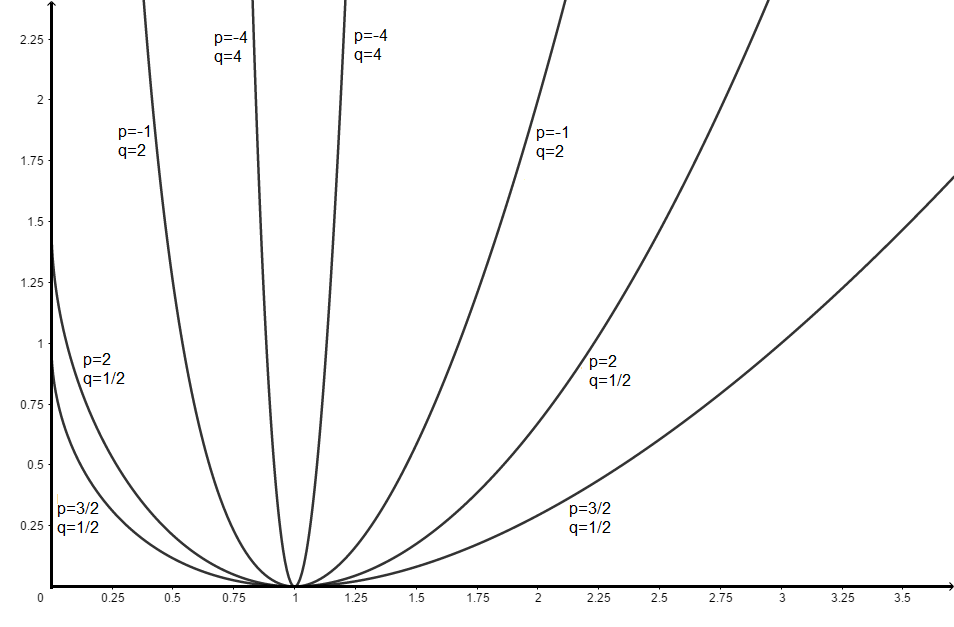}
\caption{Double-power entropies}
\end{figure}

\subsection{Divergences and triangle inequality}
As we have previously seen, starting from a function $F\in \Gamma_0(\mathbb{R}_{+})$ such that $F(s)=0$ if and only if $s=1$, the marginal perspective function $H$ is non-negative, symmetric and $H(r,t)=0$ if and only if $r=t$ (if no confusion is possible, from now on I will denote by $H$ the function $H_F$). In this section I begin the discussion regarding another property that $H$ has to fulfill in order to be a metric on $[0,\infty)$: the triangle inequality. 

When I write "$d$ is a metric on a space $X$" I mean that $d:X\times X\rightarrow [0,+\infty)$ is a function such that $d(x,y)=0$ if and only if $x=y$, it is symmetric, i.e. $d(x,y)=d(y,x)$ for every $x,y\in X$, and it satisfies the triangle inequality in the sense that $d(x,z)\leq d(x,y)+d(y,z)$ for every $x,y,z\in X$.

Since I will prove that the only divergence that is also a distance is the total variation, I will also discuss when the power $H^a$, $a\in (0,1)$, is a metric on $[0,+\infty)$. 

The convexity of the function $H$ implies that  
\begin{equation}
\label{monotonia funzione H}
H(r,t)\geq H(s,t) \ \mbox{and} \ H(r,s)\leq H(r,t) \ \mbox{for every} \ 0\leq r\leq s\leq t.
\end{equation}  

I recall this simple Lemma:
\begin{lemma}
\label{cambio metrica}
Let $(X,d)$ be a metric space and $f:[0,+\infty)\rightarrow [0,+\infty)$ be a concave function such that $f(r)=0$ if and only if $r=0$. Then $(X,f(d))$ is a metric space.
\end{lemma}
\begin{proof}
$f(d(x_1,x_2))\geq 0$ and $f(d(x_1,x_2))=0$ if and only if $d(x_1,x_2)=0$ which implies $x_1=x_2$. It is clear that $f(d)$ is symmetric. Since $f$ is concave and $f(r)>0$ for every $r>0$ it follows that $f$ is increasing and subadditive, thus
\begin{align*}
f\big(d(x_1,x_3)\big)&\leq f\big(d(x_1,x_2)+d(x_2,x_3)\big)\leq f\big(d(x_1,x_2)\big)+f\big(d(x_2,x_3)\big).
\end{align*}

\end{proof}

An easy consequence of the Lemma is that if $H^a$ is a metric, then $H^b$ is a metric for every $b\in(0,a]$.

Using the symmetry, the $1$-homogeneity of the function $H$ together with the property \eqref{monotonia funzione H}, it follows that the triangle inequality for the function $H^a$ is equivalent to the following inequality
\begin{equation}
\label{triangolare a una variabile}
H^a(u,1)\leq H^a(u,v)+H^a(v,1)=v^aH^a\Big(\frac{u}{v},1\Big)+H^a(v,1), \ \ \mathrm{for \ any} \ 0\leq u<v<1.
\end{equation}

A last useful remark is that 
\begin{equation}
\lim_{u\downarrow 0}H(u,1)<+\infty
\end{equation}
is a necessary condition for the existence of a power $a$ such that $H^a$ is a metric. 

Regarding the examples previously seen, it was proved by Kafka, Osterreicher and Vincze $\cite{Kafka}$ that $H^a_{\chi^{\alpha}}$ is a metric when $a=1/\alpha.$

The Matusita divergences clearly provide the distance $H^a_{M_a}$. 

When $p>1$, $\lim_{u\downarrow 0}H_{V_p}(u,1)=+\infty$ so that, except for the case $p=1$, the power-logarithmic entropy is not a metric for every power $a$.

I now turn the attention to the function $H_p:=H_{U_p}$. It has the following expression

\begin{equation}
\begin{cases}H_p(r,t)=\frac{2}{p}\Big[\mathfrak{M}_1(r,t)-\mathfrak{M}_{1-p}(r,t)\Big],&\mbox{if} \ p\neq 0,\\
H_0(r,t)=r\ln{r}+t\ln{t}-(r+t)\ln{\Big(\frac{r+t}{2}\Big)}
,\end{cases}
\end{equation} 
that is also valid when $rt=0$ with the convention $0\ln(0)=0$.
 
As I have already notice, $H_p$ is the square of a metric on $[0,+\infty)$ for $p=0,p=\frac{1}{2},p=1$. I investigate now the same question for every real number $p.$ This was already done by Osterreicher in the case $p\geq 1$ $\cite{Osterreicher}$. Following the same approach I prove:

\begin{theorem}
\label{triangolare costless power mean}
The induced marginal perspective function $H_p$ is the square of a metric on $[0,+\infty)$ for any $p\in (-\infty,\frac{1}{2}]\cup[1,+\infty)$. $\sqrt{H_p}$ does not satisfy the triangle inequality if $p\in (\frac{1}{2},1).$
\end{theorem}

For the proof of the Theorem I will use the following lemma. It is the first example in the paper of a fact that will be recurrent: the central role of the class of Matusita divergences in the study of the metric properties of the marginal perspective function.

\begin{lemma}
\label{lemmakafka}
Given a number $a\in (0,1]$ and an induced marginal perspective function $H$, if
$$h(u):=\frac{(1-u^a)^{\frac{1}{a}}}{H(u,1)}$$ is decreasing in $[0,1),$ then $H^a$ satisfies the triangle inequality.
\end{lemma}
\begin{proof}
Due to the monotonicity of the square root function, one has that 
$$h^a(u)=\frac{1-u^a}{H^a(u,1)}$$
is decreasing in $[0,1)$, so that $h^a(u)\geq h^a(v)$ and $h^a(u)\geq h^a(\frac{u}{v})$ if $0\leq u<v<1$. 
It follows that 
\begin{align}
H^a(u,1)&=\frac{1-u^a}{h^a(u)}=\frac{1-v^a}{h^a(u)}+\frac{v^a-u^a}{h^a(u)}\\
&\leq \frac{1-v^a}{h^a(v)}+\frac{v^a\Big(1-\big(\frac{u}{v}\big)^a\Big)}{h^a(\frac{u}{v})}=H^a(u,v)+H^a(v,1).
\end{align}
\end{proof}

\begin{proof}[Proof of Theorem \ref{triangolare costless power mean}]
Using now Lemma $\ref{lemmakafka}$, it remains to show that the function 

$$h_p(u):=\frac{(1-\sqrt{u})^2}{f_p(u)}$$ 
is decreasing in $(0,1)$, where I have used the notation $f_p(u):=H_p(u,1).$
The derivative of the function $h_p$ is the following:

\begin{equation}
h_p'(u)=-\frac{2}{p}\bigg(\frac{1}{\sqrt{u}}-1\bigg)\frac{1}{f^2_p(u)}{\phi_p(u)},
\end{equation}
where I set

\begin{equation}
\phi_p(u)=2^{-1}({u^{\frac{1}{2}}+1})-2^{-\frac{1}{1-p}}(u^{1-p}+1)^{\frac{1}{1-p}-1}(u^{\frac{1}{2}-p}+1).
\end{equation}
Note that $\phi_p(1)=0$ and $\psi_p(u)=\sqrt{u}\phi_p'(u)$ satisfies:

\begin{equation}
\psi_p(u)=\frac{1}{4}-2^{-\frac{1}{1-p}}(u^{1-p}+1)^{\frac{1}{1-p}-2}u^{-p}\Big(\frac{1+u^{1-p}}{2}-p(1-\sqrt{u})\Big).
\end{equation}

The function $\psi_p$ is such that $\psi_p(1)=0$ and 

\begin{equation}
\psi_p'(u)=2^{-\frac{1}{1-p}}p(\frac{1}{2}-p)(u^{1-p}+1)^{\frac{1}{1-p}-3}u^{-p-1}(1-\sqrt{u})(1-u^{1-p}).
\end{equation}

Now let us suppose $p>1$: I have to prove that $\phi_p$ is positive in $(0,1)$. This is implied by $\psi_p(u)<0$ in $(0,1)$ which is true because $\psi'_p(u)$ is positive in $(0,1)$.
Similar considerations can be applied to the case $p<0$ and $p\in (0,\frac{1}{2}).$

For $p\in (\frac{1}{2},1)$ one gets $\psi_p'(u)<0$ in $(0,1)$ so $\psi_p$ is positive in $(0,1)$. This implies that $\phi_p$ is negative and so $h_p$ is increasing in $(0,1)$. As a consequence, an analysis of the proof of Lemma $\ref{lemmakafka}$ shows that the triangle inequality is reversed for these values of $p$.
\end{proof}

\begin{remark}
It was proved by Osterreicher and Vajda (\cite{Osterreicher2}) that, if $p\in (\frac{1}{2},1)$, $H_p^{1-p}$ is a metric.
\end{remark}

\subsection{Marginal perspective function and convergence properties}
\label{Sezione caratterizzazione universale}
We have seen that the construction of the marginal perspective function naturally produces a symmetric divergence. In this section I will show that this is not the only feature of the minimization procedure \eqref{minimizzazione senza costo}: iterating this process I will highlight the important role of the class of Matusita divergences.
 
I define the space $\Gamma_0^s(\mathbb{R}_{+})$ as the set of functions $F\in \Gamma_0(\mathbb{R}_{+})$ such that $F$ is equals to its reverse entropy $R$.

At the beginning of section \ref{sezione senza costo} we have seen how to generate a map $T_1: \Gamma_0(\mathbb{R}_{+})\rightarrow \Gamma_0^s(\mathbb{R}_{+})$: starting from a function $F\in \Gamma_0(\mathbb{R}_{+})$, I define $T_1(F)(s):=H_F(1,s)$, where $H_F$ is the lower semicontinuous envelope of the function $\tilde{H}_F$ obtained by \eqref{minimizzazione senza costo}.
I also denote by $T_a: \Gamma_0(\mathbb{R}_{+})\rightarrow \Gamma_0^s(\mathbb{R}_{+})$ the map given by $T_a(F):=2^{\frac{1}{a}-1}T_1(F)$ for every $a\in (0,1]$.

It is clear that the two trivial entropies
\begin{equation}
F(s)\equiv 0 \ \ \ \mathrm{and}\ \ \ \   F(s)=I_{\{1\}}=\begin{cases}0 \ \  \ \ \ \mathrm{if} \ s=1 \\ +\infty \ \ \mathrm{otherwise},\end{cases}
\end{equation}
are fixed points of the map $T_a$ for any $a\in (0,1]$.

Another important property that follows immediately from the definition is that
\begin{equation}\label{monotonia T_a}
F_1\geq F_2 \implies T_a(F_1)\geq T_a(F_2).
\end{equation}

Due to the difference between the case $a=1$ and the case $0<a<1$, I have divided the analysis of the behaviour of the map $T_a$. Nevertheless, the strategy behind the proofs is in common:  I will show that, under suitable conditions, the sequence $\{T^{(n)}_a(F)\}$ is monotone and the pointwise limit is a fixed point of the map $T_a$. I then prove that $T_a(F)=F$ implies $F(s)=c|s^a-1|^{\frac{1}{a}}, \ c\in [0,+\infty]$ (in the case $c=+\infty$ I mean that $c|s^a-1|^{\frac{1}{a}}=I_{\{1\}}(s)$). 

I start with a simple Lemma that provides a crucial monotonicity property.

\begin{lemma}
\label{triangolare implica monotonia}
For any $a\in (0,1]$, if $H^a$ satisfies the triangle inequality then $T_a(F)\geq F$.
\end{lemma}
\begin{proof}
For any $s,t\in \mathbb{R}_{+}$ the convexity of the function $x\mapsto x^{\frac{1}{a}}$ yields
\begin{align*}
&2^{\frac{1}{a}-1}H(1,s)+2^{\frac{1}{a}-1}H(s,t)=\frac{1}{2}\big(2H^a(1,s)\big)^{\frac{1}{a}}+\frac{1}{2}\big(2H^a(s,t)\big)^{\frac{1}{a}}\geq \big(H^a(1,s)+H^a(s,t)\big)^{\frac{1}{a}}\geq H(1,t)=F(t).
\end{align*}
The result follows by taking the infimum of the left hand side with respect to $s$. 
\end{proof}

\begin{lemma}\label{lemma:monotonia T1}
Given a function $F\in \Gamma_0^s(\mathbb{R}_{+})$ the sequence $\{T^{(n)}_1(F)\}$ is decreasing and it converges pointwise to a fixed point of the map $T_1$. 
\end{lemma}
\begin{proof}
Since the map $s\mapsto H(1,s)+H(s,t)$ is equals to $H(1,t)$ when $s=1$ or $s=t$, it follows that $T_1(F)(s)\leq F(s)$ for any $s$. Thus, for any $s$ the sequence $T^{(n)}_1(F)(s)$ is a decreasing sequence bounded from below and thus it has a limit that I denote by $F^{\infty}(s)$, and it is clear that $F^{\infty}\in \Gamma_0^s(\mathbb{R}_{+})$. The limit function $F^{\infty}$ is a fixed point of $T_1$: with the same reasoning as at the beginning of the proof, one gets $T_1(F^{\infty})\leq F^{\infty}$; for the reverse inequality I notice that for any $1\leq s\leq t$ and any $n$ it holds
$$H^{(n)}(1,s)+H^{(n)}(s,t)\geq H^{(n+1)}(1,t),$$ where $H^{(n)}$ is the perspective function induced by $T_1^{(n)}(F)$. The result follows taking the limit with respect to $n$ and then minimizing with respect to $s$. 
\end{proof}

\begin{theorem}\label{punto fisso T_1}
The only fixed point of the map $T_1$ are the functions of the form $c|s-1|$ where $c\in [0,+\infty]$. In particular, an induced marginal perspective function $H$ is a metric on $\mathbb{R}_{+}$ if and only if $H=cM_1$, $c\in (0,+\infty)$.
\end{theorem}
\begin{proof}
It is clear that the function $cM_1$ is a fixed point of $T_1$ for any $c\in[0,+\infty]$. I show now that they are the  only fixed points: since $s\mapsto H(r,s)+H(s,t)$ is a convex function that has the same value when $s=r$ and $s=t$, $T_1(F)=F$ implies that $H(r,s)+H(s,t)=H(r,t)$ for any $r\leq s\leq t$. Using the homogeneous property of the function $H$, this is equivalent to the fact that $F(s)+sF(\frac{t}{s})=F(t)$ for any $1\leq s\leq t$. In particular, taking  any $t>2$, $s=2$ and $s=\frac{t}{2}$, it holds 
\begin{equation}
F(2)+2F\Big(\frac{t}{2}\Big)=F(t), \ \ \ \ \ F\Big(\frac{t}{2}\Big)+\frac{t}{2}F(2)=F(t).
\end{equation}
By taking the difference of the previous equations, one gets $F\big(\frac{t}{2}\big)=F(2)\big(\frac{t}{2}-1\big)$ for any $t>2$ so that, using again the homogeneity and the symmetry, $H(r,t)=F(2)|r-t|$ for any $r,t$.

In order to conclude the proof I notice that if $H$ is a metric then, using Lemma \ref{triangolare implica monotonia}, it follows $T_1(F)\geq F$. Since Lemma \ref{lemma:monotonia T1} provides the converse inequality, $F$ is a fixed point of $T_1$ and the only fixed points that induces a metric on $\mathbb{R}_{+}$ are the functions of the form $cM_1$ with $c\in (0,+\infty).$
\end{proof}

In order to deal with the case $0<a<1$ I need some preliminary results and some additional assumptions.
I start by proving that every metric of the form $H^a$, $a\in(0,1]$, is a complete metric.

\begin{lemma}
\label{H completa}
Let $F\in \Gamma_0^s(\mathbb{R}_{+})$ and let us suppose that $D=H^a$ is a metric for a number $a\in (0,1]$. Then it exists $c>0$ such that $$F(s)>c|s^a-1|^{\frac{1}{a}}$$
and $H^a$ is a complete metric.
\end{lemma}
\begin{proof}
For any $0\leq u<v<1$ I rewrite the distance between $u$ and $1$ as  
\begin{equation}
D(u,1)=\frac{g(u)}{g(v)}D(v,1)+\frac{g(u)}{g(\frac{u}{v})}D(u,v)
\end{equation}
where  $\displaystyle g(u):=\frac{F^a(u)}{1-u^a}.$ Since the triangle inequality holds, at least one of the numbers $\frac{g(u)}{g(v)}$ and $\frac{g(u)}{g(\frac{u}{v})}$ is less or equal than $1$. Choosing $u:=v^2$, it follows $g(v^2)\leq g(v)$ for any $v<1$. By contradiction let us suppose it does not exists a positive constant $c$ such that $F(s)>c|s^a-1|^{\frac{1}{a}}$, then it exists a sequence $v_n\in (0,1)$ such that $g(v_n)\rightarrow 0$. So, I can find a $\bar{v}\in (0,1)$ such that $D(0,1)=g(0)>g(\bar{v})$. On the other hand, since the sequence $w_n$ defined by $w_0=\bar{v}$, $w_n=w^2_{n-1}$ converges to $0$, by continuity of the function $g$ we have that $g(w_n)\rightarrow g(0)$ which is a contradiction since $g(0)>g(w_0)$ and $g(w_n)$ is decreasing.

Now it is easy to show that the metric $D$ is complete: since $H^a$ is a metric, $H$ is symmetric and $D(0,1)=F^a(0):=c_2<+\infty$. From the convexity of the function $F$ it follows $F^a(s)\leq c_2|s-1|^a$ so that
$$c_1M^a_a\leq D\leq c_2M^a_1.$$
The result follows using the fact thta $M^a_a, M^a_1$ are two complete metrics that induce the same convergence.
\end{proof}

Recall that, given a metric space $(X,d)$ and the interval $I=[0,1]$, a curve $\gamma:I\rightarrow X$ is a $\textit{constant speed geodesic}$ if 
\begin{equation}
d\big(\gamma(t),\gamma(t')\big)=d(\gamma(0),\gamma(1))|t-t'| \ \ \ \textit{for every} \, \, t,t'\in I.
\end{equation} 
A metric space $(X,d)$ is a geodesic space if for every pair of points $x,y\in X$ it exists a constant speed geodesic between $x$ and $y$.
A well-known fact is that a complete metric space is a geodesic space if and only if for every pair of points $x,y\in X$ it exists $z\in X$ such that $d(x,z)=d(z,y)=\frac{1}{2}d(x,y)$. The point $z$ is called $\textit{mid-point}$ between $x$ and $y$.

I am now ready to prove the analogous of Theorem \ref{punto fisso T_1} in the case $0<a<1$, under an additional assumption. 

\begin{theorem}
\label{limite è punto fisso}
Let $F\in\Gamma_0^s(\mathbb{R}_{+})$ and let us suppose that $H^a$, $a\in(0,1)$, is a distance and $T_a(F)=F$. Then $F(s)=c|s^a-1|^{\frac{1}{a}}$ for a constant $c\in (0,+\infty)$.
\end{theorem}
\begin{proof}
Since $T_a(F)=F$ one has that for any $r,t$ it exists $s$ such that
\begin{equation}
\label{identità per punto fisso}
2^{\frac{1}{a}-1}H(r,s)+2^{\frac{1}{a}-1}H(s,t)=H(r,t),
\end{equation}
Using the fact that $H^a$ is a metric and the concavity of the function $f(x)=x^a$ one gets
\begin{equation}\label{dis: da dis a ug}
H^a(r,t)\leq H^a(r,s)+H^a(s,t)=\frac{[2^{\frac{1}{a}}H(r,s)]^a+[2^{\frac{1}{a}}H(s,t)]^a}{2}\leq \Big(2^{\frac{1}{a}-1}H(r,s)+2^{\frac{1}{a}-1}H(s,t)\Big)^a.
\end{equation}
Equation \eqref{identità per punto fisso} implies the equality in the inequality \eqref{dis: da dis a ug}, in particular $H^a(r,s)=H^a(s,t).$ 

Since $r,t$ are two arbitrary points and $H^a$ is a complete metric from Lemma \ref{H completa}, it follows that $(\mathbb{R}_{+},H^a)$ is a one dimensional geodesic space, so it must be isometric to $(\mathbb{R}_{+},|\cdot|)$ (for a reference see $\cite{Burago}$, chapter $2$). 
In particular it exists $\phi:\mathbb{R}_{+}\rightarrow \mathbb{R}_{+}$ increasing and continuous such that I can write $H^a(r,t)=|\phi(t)-\phi(r)|$. From the $1$-homogeneity of the function $H$, it follows $H^a(r,t)=r^aH^a(1,\frac{t}{r})$ for $r>0$, so that 
\begin{equation}
\label{identità fondamentale per la phi}
\phi(t)-\phi(r)=r^a\Big(\phi\big(\frac{t}{r}\big)-\phi(1)\Big), \ \ t\geq r.
\end{equation}
Evaluating equation \eqref{identità fondamentale per la phi} for $t=2r$ I get 
\begin{equation}
\label{prima identità phi}
\phi(2r)-\phi(r)=r^a(\phi(2)-\phi(1)), \ \ \ r\geq 1,
\end{equation}
whereas the choice $r=2$ yields
\begin{equation}\label{eq: phi valutata in t}
\phi(t)-\phi(2)=2^a\Big(\phi\big(\frac{t}{2}\big)-\phi(1)\Big), \ \ \ t\geq 2.
\end{equation}
Now consider the previous equation with $t=2r$, it follows
\begin{equation}
\label{seconda identità phi}
\phi(2r)-\phi(r)=\phi(2)+2^a(\phi(r)-\phi(1))-\phi(r), \ \ \ r\geq 1.
\end{equation}
Using now the identities \eqref{prima identità phi} and \eqref{seconda identità phi}, it follows  
$$r^a(\phi(2)-\phi(1))=\phi(2)+2^a(\phi(r)-\phi(1))-\phi(r) \ \ \mathrm{for \ any} \ r\geq 1,$$
and I can compute $\phi(r)$ as
$$\phi(r)=\frac{\phi(2)-\phi(1)}{2^a-1}(r^a-1)+\phi(1),$$
so that $\displaystyle H^a(1,r)=(r^a-1)\frac{H^a(2,1)}{2^a-1}$ for any $r\geq 1$, which prove the theorem.
\end{proof}	
\begin{remark}
I do not know if the assumption that $H^a$ is a metric can be removed in order to obtain the same characterization as in Theorem \ref{punto fisso T_1}. The difficulty is that the value of the function $s\mapsto 2^{\frac{1}{a}-1}H(r,s)+2^{\frac{1}{a}-1}H(s,t)$ at $s=r$ and $s=t$ is strictly greater that $H(r,t)$, unless $a=1$.
\end{remark}

In order to obtain that also in the case $0<a<1$ the limit function is a fixed point of the map $T_a$, I need the following Lemma:

\begin{lemma}\label{lemma:lim-min}
Let $X$ be a compact space and let $f_n:X\rightarrow [0,+\infty]$ be a sequence of lower semicontinuous functions such that $f_n(x)\leq f_{n+1}(x)$ for every $n\in \mathbb{N}$ and every $x\in X$. 
Then 
$$\lim_{n\to \infty}\min_{x\in X}f_n(x)=\min_{x\in X}f_{\infty}(x),$$
where I put $f_{\infty}(x):=\lim_{n\to \infty}f_n(x).$
\end{lemma}
\begin{proof}
The functions $f_n$ and $f_{\infty}$ are lower semicontinuous over a compact set so that they have a minimum. 
Since $f_n(x)\leq f_{\infty}(x)$ for every $x\in X$ it is clear that
$$\lim_{n\to \infty}\min_{x\in X}f_n(x)\leq\min_{x\in X}f_{\infty}(x).$$
Let us suppose now $a<\min_{x\in X}f_{\infty}(x)$, so that for every $x\in X$ $a<f_{\infty}(x)$. Since $\lim_{n}f_n(x)=f_{\infty}(x)$, it exists $n=n(x)$ such that $a<f_n(x)$. It follows that the family $\{a<f_n\}_{n\in \mathbb{N}}$ is an open cover of $X$. Let $n_1,...,n_j$ be a finite collection of indexes such that 
$$X\subset \{a<f_{n_1}\}\cup ... \cup \{a<f_{n_j}\}.$$
Let $N:=\max\{n_1,...,n_j\}$, so that $X\subset \{a<f_N\}$ since $f_n$ are increasing. This implies that $a<f_n(x)$ for every $x\in X$ so that $a<\lim_{n\to \infty}\min_{x\in X}f_n(x)$. Since $a$ is an arbitrary number less than $\min_{x\in X}f_{\infty}(x)$, the Lemma follows.
\end{proof}

I can now state the Theorem about the convergence of the iterations of the map $T^a$.

\begin{theorem}\label{th:triangolare implica convergenza}
Let $a\in(0,1)$. Given a function $F\in \Gamma_0^s(\mathbb{R}_{+})$, if $H^a$ is a metric then the sequence $\{T_a^{(n)}(F)\}$ converges pointwise to a fixed point of the map $T_a$. In particular, if the limit function $F^{\infty}$ is such that $(H^{\infty})^a$ is a metric, then $F^{\infty}(s)=c|s^a-1|^{\frac{1}{a}}$ where $c\in (0,+\infty).$ 
\end{theorem}
\begin{proof}
Lemma \ref{triangolare implica monotonia} implies that $T_a(F)\geq F$. By the monotonicity property \eqref{monotonia T_a} the sequence $T_a^{(n)}(F)$ is increasing so it converges pointwise to a function $F^{\infty}:\mathbb{R}_{+}\rightarrow [0,\infty]$. Since $H^a$ is a metric, $F$ is convex and finite everywhere (thus continuous), as well as $T_a^{(n)}(F)$. I want to show that $F^{\infty}$ is a fixed point of $T_a$: 
\begin{align*}
&T_a(F^{\infty})(s)=sc^{-}\Big(2^{\frac{1}{a}-1}\inf_{\theta >0}\big(F^{\infty}(\theta)+\theta F^{\infty}(\frac{s}{\theta})\big)\Big)=sc^{-}\Big(2^{\frac{1}{a}-1}\lim_{n\to \infty}\inf_{\theta>0}\big(T_a^{(n)}(F)(\theta)+\theta T_a^{(n)}(F)(\frac{s}{\theta})\big)\Big)\\
&=sc^{-}\Big(\lim_{n\to \infty}T_a^{(n+1)}(F)(s)\Big)=F^{\infty}(s)
\end{align*}
where I have denoted by $sc^{-}(f)$ the lower semicontinuous envelope of the function $f$ and I have used Lemma \ref{lemma:lim-min} applied to $f_n(\theta):=T_a^{(n)}(F)(\theta)+\theta T_a^{(n)}(F)(\frac{s}{\theta})$ and $X:=[1,s]$.
The conclusion follows from Theorem \ref{limite è punto fisso}.
\end{proof}
\begin{remark}
It is not difficult to show that $F^{\infty}$ can be equal to $I_{\{1\}}$. For example, take $F(s)=|s-1|$ and consider the sequence $T_a^{(n)}(F)$ with $a\in (0,1).$
\end{remark}

In the final part of this section I want to study the connection between the behaviour of the function $F$ in a neighborhood of $1$ and the limit function $F^{\infty}$.
I start with two lemmas:

\begin{lemma}
\label{lemma:limite dall'alto}
Let $a\in (0,1]$, $b>1$, $c\in (0,+\infty)$ and $\bar{F}\in \Gamma_0^s(\mathbb{R}_{+})$ be the function defined by 
$$\bar{F}(s):=\begin{cases}c|s^a-1|^{\frac{1}{a}} \  \ s\in [\frac{1}{b},b], \\ +\infty \ \ \ \ \ \ \ \ \ \mathrm{otherwise}. \end{cases}$$
Then $\displaystyle \lim_{n\to \infty} T_a^{(n)}({\bar{F}})(s)=c|s^a-1|^{\frac{1}{a}}.$
\end{lemma}
\begin{proof}
It is sufficient to consider the case $s>1$; by definition we have
\begin{equation}\label{eq:minimizzazione per matusita}
T_a(\bar{F})(s)=2^{\frac{1}{a}-1}\inf_{\theta \in [1,s]}\bar{F}(\theta)+\theta \bar{F}\Big(\frac{s}{\theta}\Big).
\end{equation}
When $b^2<s$ it is clear that $T_a(\bar{F})(s)=+\infty$. Moreover, I notice that in the case
\begin{equation}\label{eq:condizione per minimizzante}
\mathfrak{M}_a(1,s)\leq b, \ \mathrm{and} \ \frac{s}{\mathfrak{M}_a(1,s)}\leq b,
\end{equation}
the expression \eqref{eq:minimizzazione per matusita} is minimized by $\theta=\mathfrak{M}_a(1,s)$, so that $T_a(\bar{F})(s)=c|s^a-1|^{\frac{1}{a}}$ for such an $s$. Using now the bound given by Theorem \ref{proprieta' medie}, I deduce that the inequalities \eqref{eq:condizione per minimizzante} are certainly satisfied when $1\leq s\leq 2b-1$. The theorem is now an easy consequence of the fact that the sequence $b_0:=b$, $b_{n+1}:=2b_n-1$ is strictly increasing and it diverges to $+\infty$.
\end{proof}

\begin{lemma}
\label{lemma:limite dal basso}
Let $a\in (0,1]$, $b>1$, $c\in (0,+\infty)$ and $\underbar{F}\in \Gamma_0^s(\mathbb{R}_{+})$ be the function defined by 
$\underbar{F}(s):=c|s^a-1|^{\frac{1}{a}}$ when $s\in [\frac{1}{b},b]$ and extended linearly outside in such a way that the left derivative of $\underbar{F}$ at $b$ is the slope of the linear extension in $[b,+\infty)$.
Then $\displaystyle \lim_{n\to \infty} T_a^{(n)}(\underbar{F})(s)=c|s^a-1|^{\frac{1}{a}}.$
\end{lemma}
\begin{proof}
The lemma follows if I prove that 
\begin{equation}\label{eq:triangolare per F barrato}
\underbar{F}^a(t)\leq \underbar{F}^a(s)+s\underbar{F}^a\Big(\frac{t}{s}\Big)
\end{equation}
for every $1\leq s \leq t$. Indeed \eqref{eq:triangolare per F barrato} implies that $H^a$ is a distance, so that, by Theorem \ref{th:triangolare implica convergenza}, $T_a^{(n)}(\underbar{F})$ must converge to a function $F^{\infty}$ that is a fixed point of $T_a$. 
Since $T_a^{(n)}(\underbar{F})(s)=c|s^a-1|^{\frac{1}{a}}$ for every $n$ and every $s\in [\frac{1}{b},b]$, it holds $F^{\infty}(s)=c|s^a-1|^{\frac{1}{a}}$ for every $s\in [\frac{1}{b},b]$ and this implies that $F^{\infty}(s)=c|s^a-1|^{\frac{1}{a}}$ for every $s$. Indeed, let us suppose by contradiction it exists $s_0>1$ such that $F^{\infty}(s_0)\neq c|(s_0)^a-1|^{\frac{1}{a}}$ and consider the constant $k\neq c$ such that $F^{\infty}(s_0)=k|(s_0)^a-1|^{\frac{1}{a}}.$ Since $F^{\infty}$ and $k|(s_0)^a-1|^{\frac{1}{a}}$ are fixed points of $T_a$ and they coincide in $s_0$, it must exists another number $s_1$, $1<s_1<s_0$, where they coincide. Iterating the argument it is easy to show that $F^{\infty}$ and $k|(s_0)^a-1|^{\frac{1}{a}}$ have to coincide on a sequence of numbers that converges to $1$ but this is absurd since $F^{\infty}(s)=c|s^a-1|^{\frac{1}{a}}$ for every $s\in [\frac{1}{b},b]$ and the functions $c|s^a-1|^{\frac{1}{a}}$ and $k|s^a-1|^{\frac{1}{a}}$ coincide only at $s=1.$

It remains to show that \eqref{eq:triangolare per F barrato} holds. I use Lemma \ref{lemmakafka}: I have to prove that the function
$$s\mapsto \frac{|s^a-1|^{\frac{1}{a}}}{\underbar{F}(s)}$$
is increasing in $(1,+\infty)$: this is obvious in the interval $(1,b]$; consider now two numbers $r,t$ such that $b<r<t$. I define $s\mapsto l_r(s)$ to be the affine function that coincide with $\underbar{F}$ at $b$ and such that $l_r(r)=c|r^a-1|^{\frac{1}{a}})$, and I notice that the convexity of the function $s\mapsto c|s^a-1|^{\frac{1}{a}}$ implies that the slope of $l_r$ is greater or equal than the positive slope of the function $\underbar{F}$ in $(b,+\infty)$. Using again the convexity of the function $c|s^a-1|^{\frac{1}{a}}$ and the trivial fact that the quotient 
$$s\mapsto \frac{l_r(s)}{\underbar{F}(s)}$$
is increasing in $(b,+\infty)$, I conclude because
\begin{equation}
\frac{|t^a-1|^{\frac{1}{a}}}{\underbar{F}(t)}\geq \frac{l_r(t)}{\underbar{F}(t)}\geq \frac{l_r(r)}{\underbar{F}(r)}=\frac{c|r^a-1|^{\frac{1}{a}}}{\underbar{F}(r)}.
\end{equation}
\end{proof}

\begin{theorem}
Let $F\in \Gamma_0^s(\mathbb{R}_{+})$ be a function such that 
\begin{equation}\lim_{s\to 1}\frac{F(s)}{c|s^a-1|^{\frac{1}{a}}}=1.
\end{equation}
Then 
\begin{equation}\lim_{n\to +\infty}T_a^{(n)}(F)(s)=c|s^a-1|^{\frac{1}{a}}.
\end{equation}
\end{theorem}
\begin{proof}
For every $\epsilon>0$ it exists a $b>1$ such that 
$$(1-\epsilon)c|s^a-1|^{\frac{1}{a}}\leq F(s) \leq (1+\epsilon)c|s^a-1|^{\frac{1}{a}}, \ \ \ s\in\Big[\frac{1}{b},b\Big],$$ 
so that
$$(1-\epsilon)\underbar{F}\leq F\leq (1+\epsilon)\bar{F},$$
where $\underbar{F},\bar{F}$ are defined in Lemma \ref{lemma:limite dall'alto} and \ref{lemma:limite dal basso}.
Take now an arbitrary $s\in \mathbb{R}_{+}$, from the monotonicity property \eqref{monotonia T_a} it follows 
$$(1-\epsilon)T_a^{(n)}(\underbar{F})\leq T_a^{(n)}(F)\leq (1+\epsilon)T_a^{(n)}(\bar{F}),$$
so that by Lemma \ref{lemma:limite dall'alto} and Lemma \ref{lemma:limite dal basso} one gets
$$(1-\epsilon)c|s^a-1|^{\frac{1}{a}}\leq \liminf_{n\to \infty}T_a^{(n)}(F)(s)\leq \limsup_{n\to \infty} T_a^{(n)}(F)(s)\leq (1+\epsilon)c|s^a-1|^{\frac{1}{a}}.$$
Since $\epsilon$ is arbitrary, it exists the limit of $T_a^{(n)}(F)(s)$ and it is equal to $c|s^a-1|^{\frac{1}{a}}$.

\end{proof}

\section{Marginal perspective cost}\label{sezione H con costo}
\subsection{Marginal perspective function}
In this section I introduce the marginal perspective cost. I will modify the definition of marginal perspective function that we have seen in section \ref{sezione senza costo} in order to take into account the presence of a cost function. The construction is motivated by the study of optimal entropy-transport problem (see \cite{LMS}, section $5$, and the section \ref{sezione entropia-trasporto} of the present paper).

First of all, given a number $c\in[0,+\infty)$ and an admissible entropy function $F$, the marginal perspective function $H_c:[0,\infty)\times[0,\infty)\rightarrow [0,\infty]$ is defined as the lower semicontinuous envelope of the function

\begin{equation}
\label{espressione H tilde}
\tilde{H}_c(r_1,r_2):=\inf_{\theta>0}\theta\Big(R\big(\frac{r_1}{\theta}\big)+R\big(\frac{r_2}{\theta}\big)+c\Big),
\end{equation} 
where $R$ is the reverse entropy function of $F$. Of course, the function $H_0$ coincides with the marginal perspective function $H_F$ introduced in section \ref{sezione senza costo}.
When the numbers $r_1,r_2$ are positive, the function $\tilde{H}_c$ can be also computed as

\begin{equation}
\label{espressione H tilde tramite F}
\tilde{H}_c(r_1,r_2)=\inf_{\theta>0}r_1F\Big(\frac{\theta}{r_1}\Big)+r_2F\Big(\frac{\theta}{r_2}\Big)+\theta c,
\end{equation}
or in terms of the perspective function as

\begin{align}
\tilde{H}_c(r_1,r_2)=&\inf_{\theta>0}\hat{F}(r_1,\theta)+\hat{F}(r_2,\theta)+\theta c \\ 
=&\inf_{\theta>0}\hat{R}(r_1,\theta)+\hat{R}(\theta,r_2)+\theta c.
\end{align}

For $c=+\infty$ I set 

\begin{equation}
H_{\infty}(r_1,r_2)=F(0)r_1+F(0)r_2.
\end{equation}

The following lemma, proved in $\cite{LMS}$ (lemma $5.3$), gives a dual characterization of $H_c$:

\begin{lemma}\label{l-duale H}
For every $c\geq 0$ the function $H_c$ can be represented as 

\begin{align}
\begin{split}
\label{duale H}
H_c(r_1,r_2)=\sup\{r_1\psi_1+r_2\psi_2:\psi_i\in \mathrm{D}(R^{*}), \ R^{*}(\psi_1)+R^{*}(\psi_2)\leq c\}. 
\end{split}
\end{align}

In particular, the marginal perspective function is lower semicontinuous, convex and positively $1$-homogeneous with respect to $(r_1,r_2)$, increasing and concave with respect to $c$. Moreover, $H_c$ coincides with $\tilde{H}_c$ in the interior of its domain.
\end{lemma}

\subsection{Induced marginal perspective cost}

When $c=c(x_1,x_2)$ is a function $c:X_1\times X_2\rightarrow [0,+\infty]$, the induced marginal perspective cost is the function $H:X_1\times [0,+\infty)\times X_2\times[0,+\infty)\rightarrow [0,+\infty]$ defined as
\begin{equation}
H(x_1,r_1;x_2,r_2):=H_{c(x_1,x_2)}(r_1,r_2).
\end{equation}

A particularly important case is when $X_1=X_2=X$ and $c$ is induced by a metric $d$ on $X$. 

Given a metric space $(X,d)$, I am interested in determining when the function $H$ is the power of a metric on the corresponding cone space. The latter is the space $\mathfrak{C}=Y/{\sim}$, where $Y=X\times [0,+\infty)$ and 
\begin{equation}
(x_1,r_1)\sim (x_2,r_2) \iff r_1=r_2=0 \ \mbox{or} \ r_1=r_2, x_1=x_2.
\end{equation}

It is important to highlight that the space $\mathfrak{C}$ can be endowed with a "natural" metric $d_{\mathfrak{C}}$ (see $\cite{Burago}$, Prop. $3.6.13$):
\begin{equation}
\label{metrica naturale cono}
d_{\mathfrak{C}}^2\big((x_1,r_1),(x_2,r_2)\big)=r_1^2+r_2^2-2r_1r_2\cos(d(x_1,x_2)\wedge \pi).
\end{equation}

\begin{theorem}
\label{manca la triangolare}
Let $F(s)$ be an admissible entropy function with a strict minimum at $s=1$ and let $c$ be a symmetric function such that $c(x_1,x_2)=0$ if and only if $\ x_1=x_2$. Then the induced marginal perspective cost $H$ is symmetric, non-negative and $H(x_1,r_1;x_2,r_2)=0$ if and only if $(x_1,r_1)\sim(x_2,r_2)$. In particular, $H$ is a well defined function on the cone $\mathfrak{C}$.
\end{theorem}
\begin{proof}
Since $0\in\mathrm{D}(R^{*})$ and $R^{*}(0)=0$ it is clear that $H\geq 0$.
Moreover, when $r_1=r_2=0$ it follows from the dual representation \eqref{duale H} that $H(x_1,r_1;x_2,r_2)=0$. If $(x_1,r_1)\sim(x_2,r_2)$ and $r_1=r_2>0$ then  $c(x_1,x_2)=0$ and the fact that the marginal perspective cost is null follows from the possible choice $\theta=r_1$ in the expression \eqref{espressione H tilde}. 
Since $c$ is symmetric it is clear that 
$$H(x_1,r_1;x_2,r_2)=H(x_2,r_2;x_1,r_1).$$ 
It remains to prove that $H=0$ implies $(x_1,r_1)\sim(x_2,r_2)$. Lemma \ref{omeomorfismo coniugata} and equation \eqref{coefficienti R} tell us that $R^{*}$ is an increasing homeomorphism between $(-\mathrm{aff}F_{\infty}, F(0))$ and $(-F'_{\infty},-F'_{0})$ with $R^{*}(0)=0$. Since $F$ is a convex function with a strict minimum at $s=1$, it holds $\mathrm{aff}F_{\infty}>0, \ F(0)>0, \ F'_{\infty}>0 , F'_{0}<0$. In particular, it exists a positive number $k>0$ such that the function $R^{*}$ is finite, continuous and strictly increasing in $(-k,k)$. Hence, it follows again from the representation \eqref{duale H} that $H(x_1,r_1;x_2,r_2)=0$ and $c(x_1,x_2)>0$ implies $r_1=r_2=0$. Moreover, when $c(x_1,x_2)=0$ we must have $r_1=r_2$: suppose by contradiction that $0=r_1<r_2$ (the other case is similar), in the equation \eqref{duale H} we find $-k<\psi_1<0<\psi_2<k$ such that $R_1^{*}(\psi_1)+R^{*}(\psi_2)\leq 0$, contradicting the fact $H=0$. Finally, when $H(x_1,r_1;x_2,r_2)=0$, $c(x_1,x_2)=0$ and $r_1,r_2$ are positive I can prove that $r_1=r_2$ using the fact that $\tilde{H}_0=0$ implies $r_1=r_2$ because, using now the expression \eqref{espressione H tilde tramite F}, I know that for every natural $n$ it exists $\theta_n$ such that 
$$0\leq r_1F\Big(\frac{\theta_n}{r_1}\Big)+r_2F\Big(\frac{\theta_n}{r_2}\Big)<\frac{1}{n}.$$ In particular, for $n$ large enough, $\theta_n \in [K_1,K_2]$ for some constants $0<K_1<1<K_2$, and by extracting a subsequence $\theta_{n_j}$ it follows that $\theta_{n_j}\rightarrow \bar{\theta}$. The lower semicontinuity of $F$ forces $\frac{\bar{\theta}}{r_1}=\frac{\bar{\theta}}{r_2}=1$ so that $r_1=r_2$.
\end{proof}
 
If the function $F$ has not a strict minimum at $s=1$, the induced marginal perspective cost can be null even if $r_1\neq r_2$. To see this, take $F:[0,+\infty)\rightarrow [0,+\infty)$ defined by
\begin{equation}
F(s):=\begin{cases}0 &\mbox{if} \ 0\leq s\leq 1,\\ s-1&\mbox{if}\ \ s>1,\end{cases}\end{equation}
that gives $H_0\equiv 0$, so that $H(x_1,r_1;x_2,r_2)\equiv 0$.

\section{Entropy-Transport problem}\label{sezione entropia-trasporto}

In this section I consider two discrete spaces $X_1=\{x_1^1,x_1^2,..,x_1^m\}$ and $X_2=\{x_2^1,x_2^2,..,x_2^n\}$ and I let $c:X_1\times X_2\rightarrow [0,+\infty]$ be a proper (i.e. not identically $+\infty$) cost function that I will denote by $c_{i,j}:=c(x_1^i,x_2^j).$ I will also denote by $\mathcal{M}(X_i)$ the set of finite, nonnegative measures on $X_i, \ i=1,2$ (I refer to \cite{LMS} for a more general topological setting).

Given two finite measures $\mu_i\in \mathcal{M}(X_i),$ which can be identified with vectors $(r_1,...,r_m)\in \mathbb{R}_{+}^m, (t_1,...,t_n)\in \mathbb{R}_{+}^n$ by
\begin{equation}
\mu_1=\sum_{i=1}^m r_i\delta_{x_1^i}, \ \ \ \ \mu_2=\sum_{j=1}^n t_j\delta_{x_2^j}, \ \ \ \ r_i, t_j>0,
\end{equation} 
the classical Optimal-Transport problem between $\mu_1$ and $\mu_2$ is defined as the minimization of the functional 
\begin{equation}
\mathcal{C}(\gamma):=\sum_{i,j}c_{i,j}\gamma_{i,j}
\end{equation}
with respect to any positive measure $\boldsymbol{\gamma}\in \mathcal{M}(X_1\times X_2),$ $\boldsymbol{\gamma}=\sum_{i,j}\gamma_{i,j}\delta_{(x_1^i,x_2^j)},$ that satisfies the marginal constraints

\begin{align}
\label{condizioni sulle marginali}
\begin{split}
&\sum_j \gamma_{i,j}=r_i, \ \ i=1,..,m, 
\\ 
&\sum_i \gamma_{i,j}=t_j, \ \ j=1,..,n,
\end{split}
\end{align}
a condition that forces the measures $\mu_1, \mu_2$ to have equal mass, i.e. $\sum_i r_i=\sum_j t_j.$

Optimal Entropy-Transport problems arise naturally when one tries to relax the request on the marginals \eqref{condizioni sulle marginali}.
Let $F$ be a superlinear entropy function, the Optimal Entropy-Transport problem between $\mu_1$ and $\mu_2$ is defined as the minimization of the functional 

\begin{equation}
\mathcal{E}(\boldsymbol{\gamma}|\mu_1,\mu_2):=\sum_i r_i F\Big(\frac{\sum_j\gamma_{i,j}}{r_i}\Big)+\sum_j t_j F\Big(\frac{\sum_i\gamma_{i,j}}{t_j}\Big)+\sum_{i,j}c_{i,j}\gamma_{i,j}
\end{equation}
with respect to any positive measure $\boldsymbol{\gamma}\in \mathcal{M}(X_1\times X_2),$ $\boldsymbol{\gamma}=\sum_{i,j}\gamma_{i,j}\delta_{(x_1^i,x_2^j)}.$ 

I notice that the presence of the admissible entropy functions $F$ in the cost functional $\mathcal{E}$ penalizes the measures $\boldsymbol{\gamma}$ that do not satisfy the constraints \eqref{condizioni sulle marginali} (at least when $F$ have a strict minimum at $1$), and it allows to minimize with respect any measure $\boldsymbol{\gamma}\in \mathcal{M}(X_1\times X_2).$

Given a measure $\boldsymbol{\gamma}\in \mathcal{M}(X_1\times X_2)$ such that 
\begin{align}
&\sum_{k=1}^n \gamma_{i,k}>0, \ \ i=1,..,m, \\ 
&\sum_{h=1}^m \gamma_{h,j}>0, \ \ j=1,..,n.
\end{align}
I call marginal perspective cost functional $\mathcal{H}(\mu_1,\mu_2|\boldsymbol{\gamma})$ the quantity
\begin{equation}
\mathcal{H}(\mu_1,\mu_2|\boldsymbol{\gamma}):=\sum_{i,j}H\bigg(x_1^{i},\frac{r_i}{\sum_{k=1}^n \gamma_{i,k}};x_2^{j},\frac{t_j}{\sum_{h=1}^m \gamma_{h,j}}\bigg)\gamma_{i,j}.
\end{equation}

An important result (Theorem $5.5$, \cite{LMS}) tell us that 
\begin{equation}
\mathsf{ET}(\mu_1,\mu_2):=\inf_{\boldsymbol{\gamma}\in \mathcal{M}(X_1\times X_2)}\mathcal{E}(\boldsymbol{\gamma}|\mu_1,\mu_2)=\inf_{\boldsymbol{\gamma}\in \mathcal{M}(X_1\times X_2)}\mathcal{H}(\mu_1,\mu_2|\boldsymbol{\gamma}).
\end{equation}
The advantages of the $\mathcal{H}$-formulation of the problem are based on the homogeneity of the marginal perspective cost, which allows another useful formulation of the problem on the cone space, and the intrinsic metric properties of the function $H$ (see \cite{LMS} for the special case of the Hellinger-Kantorovich distance and the rest of the present paper for other examples).

It is interesting to notice that one can recover the usual pure entropy problem in the case 
\begin{equation}X_1=X_2=X=\{x_1,...,x_m\} \ \ \mathrm{and} \  c(x_1,x_2)=\begin{cases} 0 &\mbox{if} \ x_1=x_2, \\
+\infty &\mbox{otherwise}.
\end{cases}
\end{equation}
In this case, it is not difficult to show (example E.$5$, \cite{LMS}) that, given two measures 
\begin{equation}
\mu_1=\sum_{i=1}^m r_i\delta_{x_i}, \ \ \ \mu_2=\sum_{i=1}^m t_i\delta_{x_i}, \ \ \ r_i,t_i>0,
\end{equation}
it holds
\begin{equation}
\mathsf{ET}(\mu_1,\mu_2)=\sum_{i=1}^m H_0\Big(\frac{r_i}{\gamma_i},\frac{t_i}{\gamma_i}\Big)\gamma_i=\sum_{i=1}^m f\Big(\frac{r_i}{t_i}\Big)t_i,
\end{equation}
where $\gamma_i >0, \  i=1,...,m,$ and $f(s)=H_0(s,1).$

\section{Triangle inequality in the Entropy-Transport case}

In this section I deal with the case $X_1=X_2=X$, $F=U_p$ and  $c(x_1,x_2)=d^2(x_1,x_2)$, where $d:X\times X\rightarrow [0,\infty)$ is a metric on the space $X$. I denote by $H_p$ the induced marginal perspective cost. In the case $p\neq 0,1$ it holds:

\begin{align}
H_p(x_1,r;x_2,t)=\frac{2}{p}\Big[\mathfrak{M}_1(r,t)-\mathfrak{M}_{1-p}(r,t)\bigg(1+(1-p)\frac{d^2(x_1,x_2)}{2}\bigg)_{+}^{\frac{p}{p-1}}\Big].
\end{align}

When $p=1$ or $p=0$ one gets:

\begin{align}
&H_1(x_1,r;x_2,t)=2\Big[\mathfrak{M}_1(r,t)-\mathfrak{M}_0(r,t)e^{-\frac{d^2(x_1,x_2)}{2}}\Big],
\\
&H_0(x_1,r;x_2,t)=r\ln{r}+t\ln{t}-(r+t)\ln{\Big(\frac{r+t}{2+d^2(x_1,x_2)}\Big)}.
\end{align}

From the previous section, taking $X=\{x\}$, we already know that $H_p$ cannot be the square of a metric if $p\in (\frac{1}{2},1)$.
I am going to prove that even for the case $p\leq \frac{1}{2}$ the triangle inequality fails, i.e. 
\begin{equation}
\sqrt{H_p(x_1,r;x_3,t)}>\sqrt{H_p(x_1,r;x_2,s)}+\sqrt{H_p(x_2,s;x_3,t)}
\end{equation}
for given values of $r,s,t,d(x_1,x_2),d(x_2,x_3),d(x_1,x_3)$.

If $0<p\leq \frac{1}{2}$ I choose $r=s=0, t>0$ so that

\begin{align}
&H_p(x_1,0;x_3,t)=\frac{t}{p}-\frac{2^{\frac{p}{p-1}}t}{p}\Big(1+(1-p)\frac{d^2(x_1,x_3)}{2}\Big)^{\frac{p}{p-1}},
\\
&H_p(x_1,0;x_2,0)=0,
\\
&H_p(x_2,0;x_3,t)=\frac{t}{p}-\frac{2^{\frac{p}{p-1}}t}{p}\Big(1+(1-p)\frac{d^2(x_2,x_3)}{2}\Big)^{\frac{p}{p-1}}.
\end{align}
The triangle inequality is clearly not satisfied when $$d(x_1,x_3)=d(x_1,x_2)>0, \ d(x_2,x_3)=0.$$
 
When $p=0$, I choose again $r=s=0, t>0$ so that

\begin{align}
&H_0(x_1,0;x_3,t)=t\ln{t}-t\ln{\Big(\frac{t}{2+d^2(x_1,x_3)}\Big)}
\\
&H_0(x_1,0;x_2,0)=0
\\
&H_0(x_2,0;x_3,t)=t\ln{t}-t\ln{\Big(\frac{t}{2+d^2(x_2,x_3)}\Big)}.
\end{align}
Once again, the triangle inequality fails for 
$$d(x_1,x_3)=d(x_1,x_2)>0, \ d(x_2,x_3)=0.$$

If $p<0$, I choose instead $0<r<s<t$ and
$$d(x_1,x_3)=d(x_1,x_2)>0, \ d(x_2,x_3)=0$$ 
so that

\begin{align}
&H_p(x_1,r;x_3,t)=\frac{2}{p}\Big[\mathfrak{M}_1(r,t)-\mathfrak{M}_{1-p}(r,t)\Big(1+(1-p)\frac{d^2(x_1,x_3)}{2}\Big)^{\frac{p}{p-1}}\Big],
\\
&H_p(x_1,r;x_2,s)=\frac{2}{p}\Big[\mathfrak{M}_1(r,s)-\mathfrak{M}_{1-p}(r,s)\Big(1+(1-p)\frac{d^2(x_1,x_3)}{2}\Big)^{\frac{p}{p-1}}\Big],
\\
&H_p(x_2,s;x_3,t)=\frac{2}{p}\Big[\mathfrak{M}_1(s,t)-\mathfrak{M}_{1-p}(s,t)\Big].
\end{align}
It is not difficult to see that the triangle inequality fails when $d(x_1,x_3)$ is sufficiently large, because $\media_{1-p}(r,s)<\media_{1-p}(r,t)$ and 
$$\Big(1+(1-p)\frac{d^2(x_1,x_3)}{2}\Big)^{\frac{p}{p-1}}\rightarrow +\infty$$
when $d(x_1,x_3)\rightarrow +\infty$.

Let us now move to the case $p\geq 1.$

\begin{theorem}
\label{teoremone}
Let us suppose $X_1=X_2=X$ and $c=d^2$ for a metric $d$ on $X$. Then $\sqrt{H_p}$ is a metric on the cone $\mathfrak{C}$ for every $p\geq 1$.
\end{theorem}
\begin{proof}
The proof is long so I have divided it in different steps:

$\textbf{Step 1.}$ $\mathit{ The\ only \ problem \ is \ the \ triangle \ inequality.}$

It is clear that $H_p$ is finite and I can apply Theorem $\ref{manca la triangolare}$ so that it remains to prove that the square root of $H_p$ satisfies the triangle inequality.

$\textbf{Step 2.}$ $\mathit{Change \ of \ the \ space \ part \ and \ case \ p=1.}$

I use now Lemma \ref{cambio metrica} in order to change the expression of the function $H_p$ in a more familiar one.

\begin{proposition}
$H_p$ is the square of a metric on the cone if $\bar{H}_p$ is the square of a metric on the cone for every metric $d$ on $X$, where I put
\begin{equation}
\bar{H}_p(x_1,r;x_3,t):=\mathfrak{M}_1(r,t)-\mathfrak{M}_{1-p}(r,t)\cos\Big(d(x_1,x_3)\land \frac{\pi}{2}\Big).
\end{equation}
\end{proposition}
\begin{proof}
In order to apply Lemma \ref{cambio metrica}, in the case $p>1$ I define $f_p:[0,+\infty)\rightarrow [0,\frac{\pi}{2}]$, 
$$f_p(d)=\arccos\Big[(1-(p-1)\frac{d^2}{2})_{+}^{\frac{p}{p-1}}\Big].$$
Thus, I have to show that $f_p$ is a concave function and $f_p(d)=0$ if and only if $d=0$. The second statement is obvious, for the first one I notice that it is enough to prove that the function is concave when $d\in \big(0,\sqrt{\frac{2}{p-1}}\big)$. Let us compute the second derivative:
I put 
$$g_p(d)=\Big(1-(p-1)\frac{d^2}{2}\Big)^{\frac{p}{p-1}},$$
so that
$$f_p(d)=\arccos(g_p(d)),$$
$$g'_p(d)=\frac{-pdg_p(d)}{\Big(1-(p-1)\frac{d^2}{2}\Big)},$$
$$g''_p(d)=\frac{p\Big((p+1)\frac{d^2}{2}-1\Big)g_p(d)}{\Big(1-(p-1)\frac{d^2}{2}\Big)^2}.$$
Thus

\begin{equation}
\label{condizione su f_p}
f''_p(d)=-\frac{(1-g_p(d)^2)g_p''(d)+g_p(d)g_p'(d)^2}{(1-g_p(d)^2)^{\frac{3}{2}}}
=-\frac{p\Big((p+1)\frac{d^2}{2}-1\Big)g_p(d)\Big(1-g_p(d)^2\Big)+p^2d^2g_p(d)^3}{\Big(1-g_p(d)^2\Big)^{\frac{3}{2}}\Big(1-(p-1)\frac{d^2}{2}\Big)^2}.
\end{equation}
Recalling that $d\in \big(0,\sqrt{\frac{2}{p-1}}\big)$ and $g_p(d)\in (0,1)$, the function $f_p$ is concave if and only if 
\begin{equation}\label{dis: concavita trasformazione}
(p+1)\frac{d^2}{2}-1+\Big(1-(p-1)\frac{d^2}{2}\Big)^{\frac{p}{p-1}}\Big((p-1)\frac{d^2}{2}+1\Big)\geq 0.
\end{equation}
Since $\frac{p}{p-1}>1$ it holds $(1-(p-1)\frac{d^2}{2}\Big)^{\frac{p}{p-1}}\geq 1-p\frac{d^2}{2}$ by the Bernoulli inequality, so that 
\begin{equation}
(p+1)\frac{d^2}{2}-1+\Big(1-(p-1)\frac{d^2}{2}\Big)^{\frac{p}{p-1}}\Big((p-1)\frac{d^2}{2}+1\Big)\geq (p+1)\frac{d^2}{2}-1+\Big(1-p\frac{d^2}{2}\Big)\Big((p-1)\frac{d^2}{2}+1\Big)=p\frac{d^2}{2}\Big(1-(p-1)\frac{d^2}{2}\Big)
\end{equation}
and \eqref{dis: concavita trasformazione} follows.
In the case $p=1$ I have to check that $f_1:[0,+\infty)\rightarrow [0,\frac{\pi}{2})$ defined by
$$f_1(d)=\arccos(e^{-\frac{d^2}{2}})$$
is concave and $f_1(d)=0$ if and only if $d=0$, which is trivial.

\end{proof}
It is now clear that $\bar{H}_1$ is the square of a metric on the cone space, because \eqref{metrica naturale cono} is the square of a metric on the cone space and $d\wedge \frac{\pi}{2}=(d\wedge \frac{\pi}{2})\wedge \pi$ is a metric if $d$ is a metric. 

$\textbf{Step 3.}$ $\mathit{Triangle \ inequality \ for \ large \ values \ of \ d.}$

From now on, I suppose $p>1$ and I have to show that 
\begin{equation}
\bar{H}_p(x_1,r;x_3,t)=\media_1(r,t)-\media_{1-p}(r,t)\cos\big(d(x_1,x_3)\wedge \frac{\pi}{2}\big)
\end{equation}
is the square of a metric for any metric $d$ on $X$.

\begin{lemma}
\label{crescita rispetto d}
The function
\begin{equation}
d\mapsto \sqrt{\mathfrak{M}_1(r,t)-\mathfrak{M}_{1-p}(r,t)\cos\big(d\wedge\frac{\pi}{2}\big)}
\end{equation}
is increasing in $[0,\infty)$ for $p>1$.
\end{lemma}
\begin{proof}
Just notice that 
$$d\mapsto\cos\big(d\wedge\frac{\pi}{2}\big)$$
 is decreasing in $[0,\infty)$. The result follows easily.
\end{proof}

In view of the Lemma $\ref{crescita rispetto d}$, from now on I also assume 
$$d(x_1,x_3)=d(x_1,x_2)+d(x_2,x_3),$$
and I have to prove that for every $p>1$, for every metric $d$ on $X$ and for every $r,s,t\in [0,+\infty),x_1,x_2,x_3\in X$ the following triangle inequality holds:

\begin{equation}
\label{disuguaglianza triangolare caso generale}
\sqrt{\bar{H}_p(x_1,r;x_3,t)}\leq \sqrt{\bar{H}_p(x_1,r;x_2,s)}+\sqrt{\bar{H}_p(x_2,s;x_3,t)}.
\end{equation}	

I start with the case $d(x_1,x_2)\geq\frac{\pi}{2}$ and $d(x_2,x_3)\geq\frac{\pi}{2}$. Then

\begin{equation*}
\begin{aligned}
&\bar{H}_p(x_1,r;x_3,t)=\frac{r+t}{2},
\\
&\bar{H}_p(x_1,r;x_2,s)=\frac{r+s}{2},
\\
&\bar{H}_p(x_2,s;x_3,t)=\frac{s+t}{2},
\end{aligned}
\end{equation*}

and the triangle inequality $\eqref{disuguaglianza triangolare caso generale}$ follows easily.

In the case $d(x_1,x_2)\leq\frac{\pi}{2}$ and $d(x_2,x_3)\geq\frac{\pi}{2}$ it holds

\begin{equation*}
\begin{aligned}
&\bar{H}_p(x_1,r;x_3,t)=\frac{r+t}{2},
\\
&\bar{H}_p(x_1,r;x_2,s)=\mathfrak{M}_1(r,s)-\mathfrak{M}_{1-p}(r,s)\cos\big(d(x_1,x_2)\big),
\\
&\bar{H}_p(x_2,s;x_3,t)=\frac{s+t}{2}.
\end{aligned}
\end{equation*}

In view of the Lemma $\ref{crescita rispetto d}$ the worst case is when $d(x_1,x_2)=0$, so that it is sufficient to prove

\begin{equation}
\label{scelgo p=1}
\sqrt{r+t}\leq \sqrt{r+s-2\mathfrak{M}_{1-p}(r,s)}+\sqrt{s+t}.
\end{equation}
Using now the Lemma $\ref{proprieta' medie}$, the right hand side of $\eqref{scelgo p=1}$ is not lower than 
$$\sqrt{r+s-2\sqrt{rs}}+\sqrt{s+t},$$
hence I have to prove that 
$$\sqrt{r+t}\leq |\sqrt{r}-\sqrt{s}|+\sqrt{s+t},$$
which is obvious in the case $r\leq s$, on the other hand if $r>s$ one gets

\begin{equation}
\label{da elevare al quadrato}
\sqrt{r+t}+\sqrt{s}\leq \sqrt{r}+\sqrt{s+t},
\end{equation}
and taking the square of both sides $\eqref{da elevare al quadrato}$ is trivially proved.

Now I suppose $d(x_1,x_3)\geq\frac{\pi}{2}$, $d(x_1,x_2)<\frac{\pi}{2}$ and $d(x_2,x_3)< \frac{\pi}{2}$. Then

\begin{equation*}
\begin{aligned}
&\bar{H}_p(x_1,r;x_3,t)=\frac{r+t}{2},
\\
&\bar{H}_p(x_1,r;x_2,s)=\mathfrak{M}_1(r,s)-\mathfrak{M}_{1-p}(r,s)\cos\big(d(x_1,x_2)\big),
\\
&\bar{H}_p(x_2,s;x_3,t)=\mathfrak{M}_1(s,t)-\mathfrak{M}_{1-p}(s,t)\cos\big(d(x_2,x_3)\big).
\end{aligned}
\end{equation*}

By the same reasoning as before, it is sufficient to show the inequality 

\begin{equation}
\label{caso cono d_13>pi/2}
\sqrt{\frac{r+t}{2}}\leq \sqrt{\mathfrak{M}_1(r,s)-\sqrt{rs}\cos\big(d(x_1,x_2)\big)}+\sqrt{\mathfrak{M}_1(s,t)-\sqrt{st}\cos\big(d(x_2,x_3)\big)},
\end{equation}
that follows from the triangle inequality for the cone distance $d_{\mathfrak{C}}$, since 
$$\sqrt{\frac{r+t}{2}}\leq \sqrt{\mathfrak{M}_1(r,t)-\sqrt{rt}\cos\big(d(x_1,x_3)\big)}$$
if $\frac{\pi}{2}\leq d(x_1,x_3)<\pi$.

$\textbf{Step 4.}$ $\mathit{Triangle \ inequality \ with \ }d< \frac{\pi}{2} \ \mathit{and} \ t\leq s$

Thus, I can assume
$$d(x_1,x_3)<\frac{\pi}{2}, \ d(x_1,x_2)<\frac{\pi}{2}, \ d(x_2,x_3)<\frac{\pi}{2}.$$
Without loss of generality, I can also assume $r<t$ in the inequality \eqref{disuguaglianza triangolare caso generale}, so that I have to deal with three cases: $s\leq r$, $r<s<t$, $t\leq s$. 
In this step of the proof, I start with the latter case:

\begin{lemma}
For any fixed $r,t,x_1,x_2,x_3$, the function 
\begin{equation}
s\mapsto  \sqrt{\bar{H}_p(x_1,r;x_2,s)}+\sqrt{\bar{H}_p(x_2,s;x_3,t)}
\end{equation}
is increasing in $[t,+\infty)$ .
\end{lemma}
\begin{proof}
The result follows if I prove that for any fixed $x_1,x_2$ the function 
$$f_p(u)=\bar{H}_p(x_1,1;x_2,u)$$
is increasing in $[1,+\infty)$. 
This easily follows since

\begin{align}
f'_p(u)&=\frac{1}{2}-\frac{u^{-p}}{2}\Big(\frac{1+u^{1-p}}{2}\Big)^{\frac{p}{1-p}}\cos\big(d(x_1,x_2)\big)\geq \frac{1}{2}-\frac{u^{-p}}{2}\Big(\frac{1+u^{1-p}}{2}\Big)^{\frac{p}{1-p}}>0,
\end{align}
where the last inequality holds since it is equivalent to the following
$$\mathfrak{M}_{1-p}(1,u)<u.$$
\end{proof}
Thus, it is sufficient to show the case $s<t$.

$\textbf{Step 5.}$ $\mathit{Case \ }r<s<t$

I start with a useful lemma:

\begin{lemma}
\label{da radici a senza}
Let $A,B,C$ three non-negative numbers. Then 

\begin{equation}
\label{caso con radici}
\sqrt{C}\leq \sqrt{A}+\sqrt{B}
\end{equation}
if and only if for every $\alpha,\beta \in (0,1)$ such that $\alpha + \beta=1$ we have

\begin{equation}
\label{caso al quadrato}
C\leq \frac{A}{\alpha}+\frac{B}{\beta}.
\end{equation}
\end{lemma}
\begin{proof}
Let us suppose $\eqref{caso con radici}$. Then
$$C\leq \Big(\alpha\frac{\sqrt{A}}{\alpha}+\beta\frac{\sqrt{B}}{\beta}\Big)^2\leq \frac{A}{\alpha}+\frac{B}{\beta}$$
where I have used the Jensen inequality for the convex function $f(x)=x^2$.
In order to show that $\eqref{caso al quadrato}\Rightarrow \eqref{caso con radici}$ I notice that if $A=0$ or $B=0$ the result is clearly true, otherwise I choose $\alpha,\beta$ such that $\frac{\sqrt{A}}{\alpha}=\frac{\sqrt{B}}{\beta}$. Thus
$$ \big(\sqrt{A}+\sqrt{B}\big)^2=\Big(\alpha\frac{\sqrt{A}}{\alpha}+\beta\frac{\sqrt{B}}{\beta}\Big)^2= \frac{A}{\alpha}+\frac{B}{\beta}\geq C.$$
\end{proof}

In order to simplify the notation, from now on I put $d(x_1,x_3)=d_{13}, \ d(x_1,x_2)=d_{12}, \ d(x_2,x_3)=d_{23}$. Then, I can use Lemma $\ref{da radici a senza}$ and the triangle inequality in the case $p=1$ in order to derive a new inequality. Given $\alpha,\beta \in (0,1)$ such that $\alpha + \beta=1$, one gets:

\begin{multline}	
\label{riformulazione r<s<t}
\bar{H}_p(x_1,r;x_3,t)=\bar{H}_1(x_1,r;x_3,t)+\Big[\media_0(r,t)-\media_{1-p}(r,t)\Big]cos(d_{13})\leq 
\\
\frac{\bar{H}_1(x_1,r;x_2,s)}{\alpha}+\frac{\bar{H}_1(x_2,s;x_3,t)}{\beta\
}+\Big[\media_0(r,t)-\media_{1-p}(r,t)\Big]cos(d_{13})\leq 
\\
\frac{\bar{H}_p(x_1,r;x_2,s)}{\alpha}-\frac{\Big[\media_0(r,s)-\media_{1-p}(r,s)\Big]cos(d_{12})}{\alpha}+\frac{\bar{H}_p(x_2,s;x_3,t)}{\beta}-\frac{\Big[\media_0(s,t)-\media_{1-p}(s,t)\Big]cos(d_{23})}{\beta}
\\
+\Big[\media_0(r,t)-\media_{1-p}(r,t)\Big]cos(d_{13})\leq
\frac{\bar{H}_p(x_1,r;x_2,s)}{\alpha}+\frac{\bar{H}_p(x_2,s;x_3,t)}{\beta},
\end{multline}

where the last inequality in \eqref{riformulazione r<s<t} is valid if and only if (using again Lemma $\ref{da radici a senza}$):

\begin{equation}
\label{disuguaglianza r<s<t}
\sqrt{\Big[\media_0(r,t)-\media_{1-p}(r,t)\Big]cos(d_{13})}\leq 
\sqrt{\Big[\media_0(r,s)-\media_{1-p}(r,s)\Big]cos(d_{12})}
+\sqrt{\Big[\media_0(s,t)-\media_{1-p}(s,t)\Big]cos(d_{23})}.
\end{equation}

I notice that $\cos(d_{13})\leq \cos(d_{12})\wedge \cos(d_{23})$. Thus, it is enough to prove $\eqref{disuguaglianza r<s<t}$ in the case $d_{13}=d_{12}=d_{23}=0$.
Now, I adapt the strategy used in the proof of \cite[Lemma 2]{Endres}  I put $u:=\frac{r}{s}\in (0,1)$, $\beta u:=\frac{t}{s}\in (1,+\infty)$, so that $\beta$ is a real number greater than $1$. Thus, $\frac{1}{\beta}<u<1$ and, denoted by $F(s)$ the function
\begin{equation}
F(s)=\sqrt{\Big[\media_0(r,s)-\media_{1-p}(r,s)\Big]}+\sqrt{\Big[\media_0(s,t)-\media_{1-p}(s,t)\Big]},
\end{equation}
it follows

\begin{equation}
\label{derivata lato destro}
4\sqrt{s}\frac{d}{ds}F(s)=g_p(u)+g_p(\beta u),
\end{equation}

where 
\begin{equation}
g_p(u):=\frac{\media_0(u,1)-\frac{2}{u^{1-p}+1}\media_{1-p}(u,1)}{\sqrt{\media_0(u,1)-\media_{1-p}(u,1)}}.
\end{equation}

\begin{lemma}
The function 
$$u\mapsto g_p(u)+g_p(\beta u)$$
is increasing in $(\frac{1}{\beta},1)$ with only one zero inside the interval, so that $F$ is minimized when $s=r$ or $s=t$ and the inequality $\eqref{disuguaglianza r<s<t}$ holds.
\end{lemma}
\begin{proof}
Since $g_p$ is continuous in $(0,1)$ and $(1,+\infty)$, it is enough to show that $g_p$ is increasing in $(0,1)$ and $(1,+\infty)$, and 
$$\lim_{u\to 1^{-}}g_p(u)=\sqrt{2(p-1)}, \ \ \lim_{u\to 1^{+}}g_p(u)=-\sqrt{2(p-1)}.$$
The limits are easy to compute expanding the function near $u=1$. 
When $u\in(0,1)\cup(1,+\infty)$ it follows:

\begin{equation}
\label{derivata g}
g_p'(u)=
\frac{(p-\frac{1}{2})u^{-p}\Big(\frac{u^{1-p}+1}{2}\Big)^{\frac{2p}{1-p}}-pu^{-p+\frac{1}{2}}\Big(\frac{u^{1-p}+1}{2}\Big)^{\frac{2p-1}{1-p}}+\frac{1}{2}}{2\big[\media_0(u,1)-\media_{1-p}(u,1)\big]^{\frac{3}{2}}}.
\end{equation}

The proof is complete if I show that
$$(p-\frac{1}{2})u^{-p}\Big(\frac{u^{1-p}+1}{2}\Big)^{\frac{2p}{1-p}}-pu^{-p+\frac{1}{2}}\Big(\frac{u^{1-p}+1}{2}\Big)^{\frac{2p-1}{1-p}}+\frac{1}{2}>0$$
for any $p>1$ and any positive $u$.

I put $v=\frac{u^{1-p}+1}{2}$, so that I have to prove 
$$(p-\frac{1}{2})\Big(\frac{2v-1}{v^2}\Big)^{\frac{-p}{1-p}}-p\Big(\frac{2v-1}{v^2}\Big)^{\frac{1-2p}{2(1-p)}}+\frac{1}{2}>0$$
for any $p>1$ and $v\in (\frac{1}{2},+\infty)$. Finally I put $w=\Big(\frac{2v-1}{v^2}\Big)^{\frac{1}{p-1}}\in (0,1)$ and I prove that 
$$h(w):=(p-\frac{1}{2})w^p-pw^{p-\frac{1}{2}}+\frac{1}{2}>0,$$
for any $p>1$ and $w\in (0,1)$.
To prove the last inequality, I notice that $h(1)=0$ and $h$ is a decreasing function because $$h'(w)=p(p-\frac{1}{2})w^{p-\frac{3}{2}}(\sqrt{w}-1)<0.$$ 
\end{proof}

$\textbf{Step 6.}$ $\mathit{Case \ }s\leq r$

The strategy is to use again Lemma $\ref{da radici a senza}$ and the triangle inequality for the case $p=1$, but I have to derive a different inequality with respect to the previous step. 

\begin{lemma}
I denote with $\theta_p:[0,+\infty)\times [0,+\infty)\rightarrow [0,+\infty)$ the function
$$\theta_p(r,t):=\frac{\media_{1-p}(r,t)}{\media_0(r,t)}.$$
Then $\theta_p(s,t)\leq \theta_p(r,t).$
\end{lemma}
\begin{proof}
It is sufficient to prove that $\theta_p(u,1)$ is increasing in $(0,1)$. This is easy to prove, indeed
$$\sqrt{u}\frac{d}{du}\theta_p(u,1)=\theta_p(u,1)\Big(\frac{u^{1-p}}{u^{1-p}+1}-\frac{1}{2}\Big)\geq 0.$$
\end{proof}

Let $\alpha,\beta$ be any two numbers in $(0,1)$ such that $\alpha+\beta=1$. Let us suppose, at first, $\theta_p(s,r)\leq \theta_p(r,t)$. Then

\begin{multline}
\label{disuguaglianza finale primo caso}
\bar{H}_p(x_1,r;x_3,t)=
\bar{H}_1(x_1,r;x_3,t)\theta_p(r,t)+\media_1(r,t)\big(1-\theta_p(r,t)\big)\leq 
\\
\frac{\bar{H}_1(x_1,r;x_2,s)}{\alpha}\theta_p(r,t)+\frac{\media_1(r,s)}{\alpha}\big(1-\theta_p(r,t)\big)+ 
\frac{\bar{H}_1(x_2,s;x_3,t)}{\beta}\theta_p(r,t)+\frac{\media_1(s,t)}{\beta}\big(1-\theta_p(r,t)\big)\leq 
\\
\frac{\bar{H}_p(x_1,r;x_2,s)}{\alpha}+\frac{\bar{H}_p(x_2,s;x_3,t)}{\beta},
\end{multline}

where the first inequality in \eqref{disuguaglianza finale primo caso} follows by the triangle inequality for $\sqrt{\bar{H}_1}$ and $\sqrt{\media_1}$, while the second inequality follows since $\theta_p(s,r)\leq \theta_p(r,t)$, $\theta_p(s,t)\leq \theta_p(r,t)$ and $\bar{H}_1\leq \media_1$.

It remains to investigate the case $\theta_p(s,r)>\theta_p(r,t)$.
Let us suppose 

\begin{equation}
\label{disuguaglianza finale}
\sqrt{\media_1(r,t)\big(1-\theta_p(r,t)\big)}\leq
\sqrt{\media_1(s,r)\big(1-\theta_p(s,r)\big)}+\sqrt{\media_1(s,t)\big(1-\theta_p(r,t)\big)}.
\end{equation}

Then 

\begin{multline}
\bar{H}_p(x_1,r;x_3,t)=
\bar{H}_1(x_1,r;x_3,t)\theta_p(r,t)+\media_1(r,t)\big(1-\theta_p(r,t)\big)\leq 
\\
\frac{\bar{H}_1(x_1,r;x_2,s)}{\alpha}\theta_p(r,t)+\frac{\media_1(s,r)}{\alpha}\big(1-\theta_p(s,r)\big)+ 
\frac{\bar{H}_1(x_2,s;x_3,t)}{\beta}\theta_p(r,t)+\frac{\media_1(s,t)}{\beta}\big(1-\theta_p(r,t)\big)\leq 
\\
\frac{\bar{H}_1(x_1,r;x_2,s)}{\alpha}\theta_p(s,r)+\frac{\media_1(s,r)}{\alpha}\big(1-\theta_p(s,r)\big)+ 
\frac{\bar{H}_1(x_2,s;x_3,t)}{\beta}\theta_p(r,t)+\frac{\media_1(s,t)}{\beta}\big(1-\theta_p(r,t)\big)\leq 
\\
\frac{\bar{H}_p(x_1,r;x_2,s)}{\alpha}+\frac{\bar{H}_p(x_2,s;x_3,t)}{\beta},
\end{multline}

where in the first inequality I use $\eqref{disuguaglianza finale}$, in the second I use the hypothesis $\theta_p(s,r)>\theta_p(r,t)$, in the third I reason as in the second step of the inequality \eqref{disuguaglianza finale primo caso} in order to replace $\theta_p(r,t)$ with $\theta_p(s,t)$.

Finally, the proof is complete if I prove the inequality $\eqref{disuguaglianza finale}$. Since the case $r=s$ is trivial, I put $u:=\frac{s}{r}<1$, $v:=\frac{t}{r}>1$, so that I can rewrite the inequality $\eqref{disuguaglianza finale}$ in the following equivalent way

\begin{equation}
(1+\sqrt{u})^2\frac{\media_0(u,1)-\media_{-1}(u,1)}{\media_0(u,1)-\media_{1-p}(u,1)}\leq
\frac{\sqrt{v}\big(\sqrt{u+v}+\sqrt{1+v}\big)^2}{\media_0(1,v)-\media_{1-p}(1,v)}.
\end{equation}

Now I use the estimate 
$$(\sqrt{u+v}+\sqrt{1+v}\big)^2\geq 1+4v,$$
so that it is sufficient to prove that for any $u\in (0,1)$ and any $v\in (1,+\infty)$ 

\begin{equation}
\label{disuguaglianza finale in u e v}
(1+\sqrt{u})^2\frac{\media_0(u,1)-\media_{-1}(u,1)}{\media_0(u,1)-\media_{1-p}(u,1)}\leq 
\frac{\sqrt{v}(1+4v)}{\media_0(1,v)-\media_{1-p}(1,v)}
\end{equation}

It is easy to see that the last inequality is true at least if $p\geq \frac{3}{2}$. For example, one can bound the left hand side with 
$$l(u):=(1+\sqrt{u})^2\frac{\media_0(u,1)-\media_{-1}(u,1)}{\media_0(u,1)-\media_{-\frac{1}{2}}(u,1)},$$ 
and the right hand side with 
$$r(u):=\frac{\sqrt{v}(1+4v)}{\sqrt{v}-1}.$$
Then, standard computations show that:

$$\sup_{u\in (0,1)}l(u)<\inf_{u\in(1,+\infty)} r(u).$$

If $1<p<\frac{3}{2}$ one needs precise bounds that I have found in $\cite{Kouba}$. 
The supremum of the left hand side of $\eqref{disuguaglianza finale in u e v}$ is $\frac{4}{p-1}$. For the right hand side of $\eqref{disuguaglianza finale in u e v}$ one has:

\begin{equation}
\frac{\sqrt{v}(1+4v)}{\media_0(v,1)-\media_{1-p}(v,1)}=\frac{\media_{p-1}(1,v)(1+4v)}{\media_{p-1}(1,v)-\media_{0}(1,v)}\geq \frac{\sqrt{v}(1+4v)}{\media_{p-1}(1,v)-\media_{0}(1,v)}\geq 4\frac{\media_1(1,v)-\media_0(1,v)}{\media_{p-1}(1,v)-\media_{0}(1,v)},
\end{equation}

and again using the results in $\cite{Kouba}$ it is proved that the sharp lower bound for the last expression is $\frac{4}{p-1}$.

\end{proof}

\noindent\textbf{Acknowledgment.} The author thanks Prof.~Giuseppe Savar\'e for many valuable suggestions.

\bibliographystyle{IEEEtran}
\bibliography{IEEEabrv,mybibfile}

\begin{thebibliography}{10}
\providecommand{\url}[1]{#1}
\csname url@samestyle\endcsname
\providecommand{\newblock}{\relax}
\providecommand{\bibinfo}[2]{#2}
\providecommand{\BIBentrySTDinterwordspacing}{\spaceskip=0pt\relax}
\providecommand{\BIBentryALTinterwordstretchfactor}{4}
\providecommand{\BIBentryALTinterwordspacing}{\spaceskip=\fontdimen2\font plus
\BIBentryALTinterwordstretchfactor\fontdimen3\font minus
  \fontdimen4\font\relax}
\providecommand{\BIBforeignlanguage}[2]{{%
\expandafter\ifx\csname l@#1\endcsname\relax
\typeout{** WARNING: IEEEtran.bst: No hyphenation pattern has been}%
\typeout{** loaded for the language `#1'. Using the pattern for}%
\typeout{** the default language instead.}%
\else
\language=\csname l@#1\endcsname
\fi
#2}}
\providecommand{\BIBdecl}{\relax}
\BIBdecl

\bibitem{Csiszar}
I.~Csiszar, ``Eine informationstheoretische ungleichung und ihre anwendung auf
  den beweis der ergodizitat von markoffschen ketten,'' \emph{Magyar. Tud.
  Akad. Mat. Kutato Int. Kozl}, vol.~8, pp. 85--108, 1963.

\bibitem{Ali}
S.~Ali and S.~Silvey, ``A general class of coefficients of divergence of one
  distribution from another,'' \emph{J. Roy. Stat. Soc. Ser. B}, vol.~28, pp.
  131--142, 1966.

\bibitem{Liese2}
F.~Liese and I.~Vajda, ``On divergences and informations in statistics and
  information theory,'' \emph{IEEE Transactions on Information Theory},
  vol.~52, pp. 4394--4412, 2006.

\bibitem{Vajda}
I.~Vajda, \emph{Theory of statistical inference and information}.\hskip 1em
  plus 0.5em minus 0.4em\relax Springer Netherlands, 1989.

\bibitem{Vajda2}
------, ``$\chi^{\alpha}$–divergence and generalized fisher’s
  information,'' in \emph{Transactions of the Sixth Prague Conference on
  Information Theory, Statistical Decision Function, Random Processes}, 1973,
  pp. 873--886.

\bibitem{Matusita}
K.~Matusita, ``Distances and decision rules,'' \emph{Annals of the Institute of
  Statistical Mathematics}, vol.~16, pp. 305--320, 1964.

\bibitem{Hellinger}
E.~Hellinger, ``Neue begründung der theorie quadratischer formen von
  unendlichvielen veränderlichen,'' \emph{J. Reine Angew. Math}, vol. 136,
  1909.

\bibitem{Lin}
J.~Lin, ``Divergence measures based on the shannon entropy,'' \emph{IEEE
  Transactions on Information Theory}, vol.~37, pp. 145--151, 1991.

\bibitem{Kullback}
S.~Kullback and R.~Leibler, ``On information and sufficiency,'' \emph{Annals of
  Mathematical Statistics}, vol.~22, 1951.

\bibitem{Csiszar2}
I.~Csiszar, ``Information-type measures of difference of probability
  distributions and indirect observation,'' \emph{Studia Scientiarum
  Mathematicarum Hungarica}, vol.~2, pp. 229--318, 1967.

\bibitem{Endres}
D.~Endres and J.~Schindelin, ``A new metric for probability distributions,''
  \emph{IEEE Transactions on Information Theory}, vol.~7, 2003.

\bibitem{Kafka}
P.~Kafka, F.~Osterreicher, and I.~Vincze, ``On powers of f-divergences defining
  a distance,'' \emph{Studia Sci. Math. Hungar.}, vol.~26, 1991.

\bibitem{Osterreicher}
F.~Osterreicher, ``On a class of perimeter-type distances of probability
  distributions,'' \emph{Kybernetika}, vol.~32, pp. 389--393, 1996.

\bibitem{Osterreicher2}
F.~Osterreicher and I.~Vajda, ``A new metric divergences on probability spaces
  and its applicability in statistics,'' \emph{Annals of the Institute of
  Statistical Mathematics}, vol.~55, pp. 639--653, 2003.

\bibitem{LMS}
M.~Liero, A.~Mielke, and G.~Savaré, ``Optimal entropy-transport problems and a
  new hellinger–kantorovich distance between positive measures,''
  \emph{Inventiones Mathematicae}, 2017.

\bibitem{Vincze}
I.~Vincze, \emph{On the Concept and Measure of Information Contained in an
  Observation}.\hskip 1em plus 0.5em minus 0.4em\relax Academic Press, 1981,
  pp. 207--214.

\bibitem{LeCam}
L.~L. Cam, \emph{Asymptotic Methods in Statistical Decision Theory}.\hskip 1em
  plus 0.5em minus 0.4em\relax Springer, 1986.

\bibitem{Bullen}
P.~S. Bullen, \emph{Handbook of Means and Their Inequalities}.\hskip 1em plus
  0.5em minus 0.4em\relax 3300 AA Dordrecht, the Netherlands: Kluwer Academic
  Publishers, 2003.

\bibitem{Liese}
F.~Liese and I.~Vajda, \emph{Convex Statistical Distances}.\hskip 1em plus
  0.5em minus 0.4em\relax Teubner, 1987.

\bibitem{Burago}
D.~Burago, S.~Burago, and S.~Ivanov, \emph{A course in metric geometry}.\hskip
  1em plus 0.5em minus 0.4em\relax Providence, RI: American Matematical
  Society, 2001.

\bibitem{Kouba}
O.~Kouba, ``Bounds for the ratios of differences of power means in two
  arguments,'' \emph{Mathematical Inequalities and Applications}, vol.~17,
  no.~3, 2014.

\end{thebibliography}

\end{document}